\documentclass[11pt, oneside]{article}   	
\usepackage{geometry}                		
\usepackage{amsmath,amsthm, amssymb}
\usepackage{thmtools}
\usepackage{graphicx}
\usepackage{enumerate}
\usepackage{enumitem}
\usepackage{float}
\usepackage{amsfonts}
\usepackage{latexsym,lineno}
\usepackage{listings}
\usepackage{hyperref}
\usepackage{tikz}
\usepackage{tikz-cd}
\usepackage{tikzit}
\usepackage[normalem]{ulem}
\usepackage{cite}
\usepackage{cleveref}
\usetikzlibrary{positioning}
\tikzstyle{Vertex}=[fill={rgb,255: red,142; green,142; blue,142}, draw=black, shape=rectangle]
\tikzstyle{white vertex}=[fill=white, draw=black, shape=circle]
\tikzstyle{grey vertex}=[fill={rgb,255: red,142; green,142; blue,142}, draw=black, shape=circle]
\tikzstyle{black vertex}=[fill=black, draw=black, shape=circle]
\tikzstyle{small black}=[fill=black, draw=black, shape=circle, scale=.5pt, tikzit draw=white]
\tikzstyle{black empty}=[fill=white, draw=black, shape=circle]
\tikzstyle{small black empty}=[fill=white, draw=black, shape=circle, scale=.5pt, tikzit draw={rgb,255: red,156; green,156; blue,156}, tikzit fill=white]

\tikzstyle{double blue edge}=[->, draw={rgb,255: red,6; green,118; blue,255}, latex-latex]
\tikzstyle{automorphism}=[->, -latex, draw=red]
\tikzstyle{new edge style 0}=[->, draw={rgb,255: red,6; green,118; blue,255}, -latex]
\tikzstyle{Edge}=[-]
\tikzstyle{Arc}=[->, -latex]
\tikzstyle{blue_edge}=[
  -, line width=1.2pt,
  draw={rgb,255: red,6; green,118; blue,255},
  postaction={decorate},
  decoration={
    markings,
    mark=at position .5 with {\node[fill=none, text={rgb,255: red,6; green,118; blue,255}, inner sep=1pt] {\tiny $\bullet$};}
  }
]
\tikzstyle{edge red}=[-, line width=1.2pt, fill=none, draw=red]
\tikzstyle{green edge}=[-, line width=1.2pt, draw={rgb,255: red,0; green,168; blue,0}, postaction={decorate},decoration={
    markings,
    mark=at position .42 with {\node[fill=none, text={rgb,255: red,0; green,168; blue,0}, inner sep=1pt] {\tiny$\bullet $};},
    mark=at position .58 with {\node[fill=none, text={rgb,255: red,0; green,168; blue,0}, inner sep=1pt] {\tiny$\bullet $};},
  }]
\tikzstyle{big arc}=[->, line width=1.2pt, -latex]
\tikzstyle{shape_border}=[-, fill={rgb,255: red,211; green,211; blue,211}, line width=0.5pt]
\tikzstyle{dashed_edge}=[-, dashed, thick]

\usetikzlibrary{decorations.markings}
\tikzset{negated/.style={
        decoration={markings, mark= at position 0.5 with {\node[transform shape] (tempnode) {$\backslash$};}}
                ,postaction={decorate}}}

\newcommand{\bb}{\mathbb}

\newcommand{\G}{\Gamma}
\newcommand{\Ra}{\Rightarrow}
\newcommand{\La}{\Leftarrow}
\newcommand{\ra}{\rightarrow}
\newcommand{\la}{\leftarrow}

\newcommand{\lra}{\leftrightarrow}
\newcommand{\da}{\mathclose{\downarrow}}

\newcommand{\set}[1]{\left\{#1\right\}}
\renewcommand{\cal}[1]{\mathcal{#1}}

\newcommand{\eps}{\epsilon}

\newcommand{\adj}{\sim}
\newcommand{\stack}[2]{\stackrel{#1}{#2}}
\newcommand{\comment}[1]{}
\newcommand{\arc}[1]{ \xrightarrow{#1}}
\newcommand{\locsim}[1]{ \stackrel{#1}{\equiv}}
\newcommand{\rem}{\mathop{Rem}}
\newcommand{\remgraph}{\mathop{RemGraph}}
\newcommand{\tree}{\mathrm{tree}} %\mathrm{tree}

\newcommand{\tra}[1][]{\tikz[<-]{
\draw[-{Triangle[angle=60:4.5pt]}, line width=1.2pt] (0,0) -- (1em,0);
\node[above] at (0.6em,0) {\scriptsize #1};
}}
\newcommand{\tla}[1][]{\tikz[<-]{
\draw[{Triangle[angle=60:4.5pt]}-, line width=1.2pt] (0,0) -- (1em,0);
\node[above] at (0.6em,0) {\scriptsize #1};
}}
\newcommand{\longtra}[1][]{\tikz[<-]{
\draw[-{Triangle[angle=60:4.5pt]}, line width=1.2pt] (0,0) -- (1.5em,0);
\node[above] at (0.6em,0) {\scriptsize #1};
}}
\newcommand{\longtla}[1][]{\tikz[<-]{
\draw[{Triangle[angle=60:4.5pt]}-, line width=1.2pt] (0,0) -- (1.5em,0);
\node[above] at (0.6em,0) {\scriptsize #1};
}}

\newtheorem{theorem}{Theorem}

\newtheorem{lemma}[theorem]{Lemma}
\newtheorem{corollary}[theorem]{Corollary}

\newtheorem{proposition}[theorem]{Proposition}

\newtheorem{definition}[theorem]{Definition}

\title{Graph-based automata}
\author{Cyril Pujol }
\date{\today}

\begin{document}
	
	\maketitle
	
	\begin{abstract}
	
		We study graph-based automata: nondeterministic finite automata obtained from edge-colored or oriented graphs by taking every vertex as both initial and accepting, and every edge as a pair of opposite transitions. The language of these automata corresponds to the set of edge-colored or oriented paths mapping to their corresponding graphs. We develop an analogous notion for trees and characterise the languages recognised by these models.
		For tree languages we prove the existence of a unique size-minimal graph and, more generally, a homomorphism-minimal graph for both word and tree languages using duality methods.
		
		In order to further motivate these models, we showcase a few results at the intersection between graph theory and automata theory:
		We relate graph-based automata to reversible automata, give a decomposition of graph-based languages into reversible languages, and introduce the remanent language of an undirected graph as the intersection of all its orientations. This remanent language captures structural information on the graph such as chromatic number.
	\end{abstract}
	
	\section{Introduction}
		A graph is a system of states together with a system of symmetric relations between them. An automaton is a transition system between states. Therefore, they are conceptually very similar structures. However, automata theory and graph theory have led these objects to evolve in very different directions. The heart of automata theory has created many strong equivalences between languages, monoids, regular expressions and logic, which allow regular languages to be extremely well understood and easy to manipulate. On the other hand, a big part of graph theory consists of considering \textsf{NP}-hard problems that need to be restricted to subclasses to be manageable. For instance, one may restrict to graphs of bounded treewidth or to perfect graphs when computing the chromatic number. It is rather uncommon to obtain characterisations that are as well behaved as regular languages. In this work, we try to take advantage of the tools of automata theory to get structural information about a graph.
		
		A graph morphism is a function between the vertices of two graphs which preserves adjacency. It is a central notion in graph theory, as it encapsulates many parameters, such as the chromatic number and the clique number. Given a fixed graph $G$, by counting the number of homomorphisms from any graph $H$ to $G$, we obtain an infinite vector known as the homomorphism profile of $G$. In \cite{L67} Lovasz proved that two graphs have the same homomorphism profile if and only if they are isomorphic.
		
		By looking only at the existence of a homomorphism instead of the exact count, we obtain the set of graphs $H$ that map into $G$. Two graphs $G$,$G'$ lead to the same set if and only if they are homomorphically equivalent, meaning that there exists a morphism from $G$ to $G'$ and a morphism from $G'$ to $G$. In each homomorphism-equivalence class of graphs, there exists a unique minimal graph (in terms of the number of vertices and edges), called the core, to which any other one maps surjectively. 
		
		By viewing the set of graphs that map into $G$ as a language associated to $G$, it is natural to draw a parallel to the existence of a minimal automaton in the context of deterministic finite automata. This similarity is accentuated by the fact that an automaton accepting a word $w$ corresponds to the existence of a function from the word $w$, seen as a path, to the states of the automaton, in a way that is coherent with the transition labels. Therefore, it is essentially a morphism from the word to the automaton.
		\vspace{10pt}
		
		We introduce the notion of graph-based automata, which are nondeterministic finite automata whose language is the set of paths mapping to an edge-colored or oriented graph. In terms of automata, this is equivalent to asking that every state be initial and terminal, and that every transition be bidirectional (in the edge-colored case), or part of a pair of opposite transitions (in the oriented case). This notion should not be confused with graph automata appearing in the literature \cite{BK08}, where automata are used to process graph-structured inputs. 
		
		We also introduce the notion of graph-based tree automata, the analogous model where paths are replaced with trees. 
		Those notions have been used before in graph theory, in particular, in works related to homomorphism, either explicitly \cite{EPTT17} %Regular families of forests, antichains and duality pairs of relational structures
		or implicitly \cite{NP25}, through the study of ``bounded path''. % On core of categorical product of (di)graphs
		But they were used only as tools and not objects of interest. 
		
		From the point of view of automata theory, to the best of our knowledge, the class of graph-based automata has not been studied as a subject in its own right. However, its defining features are very natural: asking that every state be both initial and accepting is standard, and so are bidirectional transitions. For this reason, we do not claim novelty for the model itself, but rather for its explicit formulation and study. 
		One of the closest related notion appears to be that of reversible automata and reversible languages: automata where the transition function is injective. A link between the two notions is established in \Cref{prop:reversible-proper}.

\vspace{5pt}

		In this work, we characterise the languages associated with graph-based (tree) automata and study, for a given language, the possible minimal graph-based (tree) automata associated with it. For graph-based tree automata, we prove the existence of a unique minimal automaton with respect to size. This is particularly surprising since the related automata are highly nondeterministic. There also exists a minimal automaton with respect to the homomorphism order, this is a consequence of a very general duality result by Erdös, Pálvölgyib, Tardif and Tardos \cite{EPTT17}. %Regular families of forests, antichains and duality pairs of relational structures
We give a different proof of that result, adapted to our context.
		For graph-based automata, there is no minimal graph with respect to size, but a unique homomorphism-minimal graph does exist, as a direct consequence of a well-known result by Christian Carrez \cite{ADN92}.
%A note about minimal non-deterministic automata 

	Finally, we give some applications at the intersection of automata and graph theory: we give a way to disambiguate a graph-based language by lifting it into a reversible language, and introduce the notion of the remanent language of an undirected graph, which is the language of oriented paths that can map to any orientation of the graph. This language encapsulates considerable structural information about the graph, such as the chromatic number.

	\section{Graph-based automata: Preliminaries}
	\subsection{Definition and notations}
	A graph $G$ is a pair $(V,E)$ of vertices and edges. A graph is said to be \emph{oriented} if each edge has a beginning and an end, it is said to be \emph{edge-colored} if each edge is labelled by a set of \emph{colors}. 
	Except if explicitly mentioned, we allow loops and multi-edges, and graphs are finite.
	  
	In this paper, we will mostly consider graphs that are oriented or edge-colored (or both). To have a unified vocabulary, we may refer to them as \emph{decorated graphs}, or simply as graphs in a context where the additional edge information is not the main focus. Conversely, we say that a graph is \emph{plain} to emphasise the fact that it is neither oriented nor edge-colored.
	
	In a plain graph, we use $u \adj v$ to write that $u$ and $v$ are adjacent. For a decorated graph, we write $u \arc c v$ if $u$ is connected to $v$ via the color $c$. If the graph is not oriented, then $u \arc c v$ is equivalent to $v \arc c u$, and if the graph is uncoloured, we may interpret the graph as having $c = 0$. We write $N(u)$ for the neighbours of $u$. In a decorated graph, we write $N^+_c(u)$ for the out-neighbours of $u$ via color $c$, and $N^-_c$ for the in-neighbours, then $N(u) = \bigcup_c N_c^-(u) \cup N_c^+(u)$.
	
	A \emph{graph morphism} or \emph{homomorphism} is a function between the vertex sets of two graphs that preserves adjacency. In the context of decorated graphs, a homomorphism must preserve colours and orientation: $G \xrightarrow \phi H$ is a homomorphism if and only if $\forall u,v \in V(G), u \arc c v \Ra \phi(u) \arc c \phi(v)$. We write $G \ra H$ to indicate that there exists a homomorphism from $G$ to $H$. 
	In the context of graph homomorphisms, a \emph{duality pair} (defined by Nešetřil and Pultr in \cite{NP78})  is a pair of graphs, or more generally a pair of sets of graphs $\mathcal{A},\mathcal{B}$, such that for any graph $X$, $$ \exists A \in \mathcal{A}, X \ra A \iff  \forall B \in \mathcal{B}, B \not \ra X$$
	\vspace{10pt}
	
	Consider a decorated graph $G = (V(G),E(G))$ with $E(G) \subset V(G) \times V(G) \times C$ ($C$ is the set of colors). We construct the automaton based on $G$ as the following nondeterministic finite automaton (NFA; see \Cref{fig:ex_graph_automata}):
	 
	$A(G) = (Q,\Sigma,\delta,I,F)$ with : 
	\begin{itemize}[itemsep=0pt]
		\item $Q = I = F = V(G)$
		\item $\Sigma = \set{\tra[c] | c \in C}\cup\set{\tla[c] | c \in C}$. In particular, if $G$ is only oriented, we may use $\Sigma = \set{\tra, \tla}$ and if $G$ is only edge-colored, $\tra[c]$ and $\tla[c]$ are equivalent and we may use $\Sigma = C$
		\item $\delta(u,\tra[c]) = \set{ v | u \arc c v \text{ in G}}$, \ \ $\delta(v,\tla[c]) = \set{ u | u \arc c v \text{ in G}}$
	\end{itemize}

\begin{figure}[ht]
		\centering
			\begin{tikzpicture}
	\begin{pgfonlayer}{nodelayer}
		\node [style=black empty] (0) at (-1.5, 4) {};
		\node [style=black empty] (1) at (1.5, 4) {};
		\node [style=black empty] (2) at (0, 1.5) {};
		\node [style=black empty] (3) at (-1.5, -1) {};
		\node [style=black empty] (4) at (0, -3.5) {};
		\node [style=black empty] (5) at (1.5, -1) {};
	\end{pgfonlayer}
	\begin{pgfonlayer}{edgelayer}
		\draw [style=edge red] (0) to (1);
		\draw [style={blue_edge}] (0) to (2);
		\draw [style={blue_edge}] (2) to (1);
		\draw [style={blue_edge}, in=180, out=90, loop] (0) to ();
		\draw [style=edge red, in=90, out=0, loop] (1) to ();
		\draw [style=new edge style 0, bend left=15, postaction={decorate}, decoration={  markings, mark=at position .5 with {\node[fill={rgb,255: red,6; green,118; blue,255}, inner sep=1pt] {};}}] (4) to (3);
		\draw [style=new edge style 0, bend left=15, postaction={decorate}, decoration={  markings, mark=at position .5 with {\node[fill={rgb,255: red,6; green,118; blue,255}, inner sep=1pt] {};}}] (3) to (4);
		\draw [style=new edge style 0, bend left=15, postaction={decorate}, decoration={  markings, mark=at position .5 with {\node[fill={rgb,255: red,6; green,118; blue,255}, inner sep=1pt] {};}}] (4) to (5);
		\draw [style=new edge style 0, bend left=15, postaction={decorate}, decoration={  markings, mark=at position .5 with {\node[fill={rgb,255: red,6; green,118; blue,255}, inner sep=1pt] {};}}] (5) to (4);
		\draw [style=new edge style 0, in=-180, out=90, loop, postaction={decorate}, decoration={  markings, mark=at position .5 with {\node[fill={rgb,255: red,6; green,118; blue,255}, inner sep=1pt] {};}}] (3) to ();
		\draw [style=automorphism, bend right=15] (3) to (5);
		\draw [style=automorphism, bend left=345] (5) to (3);
		\draw [style=automorphism, in=90, out=0, loop] (5) to ();
	\end{pgfonlayer}
\end{tikzpicture}
			\hspace{15pt}
			\begin{tikzpicture}
	\begin{pgfonlayer}{nodelayer}
		\node [style=black empty] (0) at (-2, 4) {};
		\node [style=black empty] (1) at (1, 4) {};
		\node [style=black empty] (2) at (-0.5, 1.5) {};
		\node [style=black empty] (3) at (-2, -1) {};
		\node [style=black empty] (4) at (1, -1) {};
		\node [style=black empty] (5) at (-0.5, -3.5) {};
		\node [style=none] (6) at (1, -2.5) {$\longtra$};
		\node [style=none] (7) at (-2, -2.5) {$\longtra$};
		\node [style=none] (8) at (-0.5, -0.25) {$\longtra$};
		\node [style=none] (9) at (-0.5, -1.5) {$\longtla$};
		\node [style=none] (10) at (-1, -2) {$\longtla$};
		\node [style=none] (11) at (0, -2) {$\longtla$};
	\end{pgfonlayer}
	\begin{pgfonlayer}{edgelayer}
		\draw [style=Arc] (0) to (2);
		\draw [style=Arc] (2) to (1);
		\draw [style=Arc] (0) to (1);
		\draw [style=Arc, bend right=15] (3) to (5);
		\draw [style=Arc, bend right=15] (5) to (4);
		\draw [style=Arc, bend left=15] (3) to (4);
		\draw [style=Arc, bend right=15] (5) to (3);
		\draw [style=Arc, bend right=15] (4) to (5);
		\draw [style=Arc, bend left=15] (4) to (3);
	\end{pgfonlayer}
\end{tikzpicture}
		\caption{A $2$-edge-colored graph (left) and an oriented graph (right) above their respective automata. Indicators of initial and final states are omitted. Blue edges are marked with a dot for added readability}
		\label{fig:ex_graph_automata}
\end{figure}
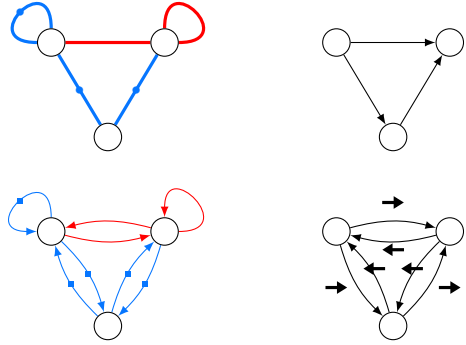

We write $L(G)$ for the language of the automaton based on $G$.
Note that words over the alphabet $\set{\tra[c] | c \in C}\cup\set{\tla[c] | c \in C}$ can be interpreted as oriented and/or edge-labelled paths. Conversely, any path corresponds to two words over this alphabet, depending on which endpoint one starts with.

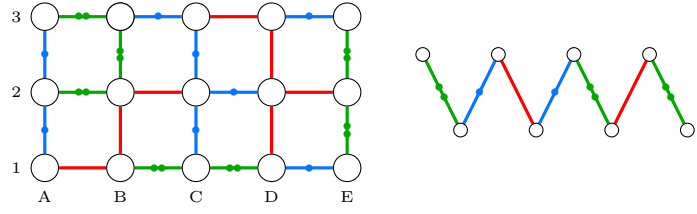
\begin{figure}[t]
		\centering
			\begin{tikzpicture}
	\begin{pgfonlayer}{nodelayer}
		\node [style=black empty] (0) at (-2, 2) {};
		\node [style=black empty] (1) at (0, 2) {};
		\node [style=black empty] (2) at (-2, 0) {};
		\node [style=black empty] (3) at (0, 0) {};
		\node [style=black empty] (4) at (0, 2) {};
		\node [style=black empty] (5) at (2, 2) {};
		\node [style=black empty] (6) at (2, 0) {};
		\node [style=black empty] (7) at (2, -2) {};
		\node [style=black empty] (8) at (0, -2) {};
		\node [style=black empty] (9) at (-2, -2) {};
		\node [style=black empty] (10) at (4, 2) {};
		\node [style=black empty] (11) at (4, 0) {};
		\node [style=black empty] (12) at (4, -2) {};
		\node [style=black empty] (13) at (6, 2) {};
		\node [style=black empty] (14) at (6, 0) {};
		\node [style=black empty] (15) at (6, -2) {};
		\node [style=small black empty] (16) at (8, 1) {};
		\node [style=small black empty] (17) at (9, -1) {};
		\node [style=small black empty] (18) at (10, 1) {};
		\node [style=small black empty] (19) at (11, -1) {};
		\node [style=small black empty] (20) at (12, 1) {};
		\node [style=small black empty] (21) at (13, -1) {};
		\node [style=small black empty] (22) at (14, 1) {};
		\node [style=small black empty] (23) at (15, -1) {};
		\node [style=none] (24) at (-2.75, 2) {\tiny 3};
		\node [style=none] (25) at (-2.75, 0) {\tiny 2};
		\node [style=none] (26) at (-2.75, -2) {\tiny 1};
		\node [style=none] (27) at (-2, -2.75) {\tiny A};
		\node [style=none] (28) at (0, -2.75) {\tiny B};
		\node [style=none] (29) at (2, -2.75) {\tiny C};
		\node [style=none] (30) at (4, -2.75) {\tiny D};
		\node [style=none] (31) at (6, -2.75) {\tiny E};
	\end{pgfonlayer}
	\begin{pgfonlayer}{edgelayer}
		\draw [style={blue_edge}] (0) to (2);
		\draw [style=edge red] (9) to (8);
		\draw [style=edge red] (8) to (3);
		\draw [style=edge red] (3) to (6);
		\draw [style={blue_edge}] (2) to (9);
		\draw [style=green edge] (8) to (7);
		\draw [style={blue_edge}] (6) to (5);
		\draw [style={blue_edge}] (5) to (4);
		\draw [style=green edge] (2) to (3);
		\draw [style=green edge] (3) to (4);
		\draw [style=green edge] (4) to (0);
		\draw [style={blue_edge}] (6) to (7);
		\draw [style=green edge] (7) to (12);
		\draw [style=green edge] (14) to (15);
		\draw [style=green edge] (13) to (14);
		\draw [style={blue_edge}] (6) to (11);
		\draw [style={blue_edge}] (12) to (15);
		\draw [style={blue_edge}] (10) to (13);
		\draw [style=edge red] (5) to (10);
		\draw [style=edge red] (10) to (11);
		\draw [style=edge red] (11) to (12);
		\draw [style=edge red] (11) to (14);
		\draw [style=green edge] (16) to (17);
		\draw [style=green edge] (23) to (22);
		\draw [style=green edge] (20) to (21);
		\draw [style={blue_edge}] (17) to (18);
		\draw [style={blue_edge}] (19) to (20);
		\draw [style=edge red] (18) to (19);
		\draw [style=edge red] (21) to (22);
	\end{pgfonlayer}
\end{tikzpicture}
		\caption{A 3-edge colored graph and a path mapping to it, finding the correct mapping is left to the reader as a puzzle\protect\footnotemark. }
		\label{fig:puzzle}
\end{figure}

\begin{definition}
	If $a\in \set{\tra[c] | c \in C}\cup\set{\tla[c] | c \in C}$, let $a^{-1}$ be the arrow in opposite direction to $a$. We extend the notion for words with $\epsilon^{-1} = \epsilon$ (the empty word) and $(aw)^{-1} = w^{-1}a^{-1}$.
\end{definition}
For undirected graphs, the inverse of a word is equal to its mirror.

\subsection{Basic properties}

We start with a few basic properties of graph-based automata.
\begin{proposition}\label{prop:basic-prop}
	If $G$, $H$ are decorated graphs, then the following are true:
	\begin{enumerate}
		\item $L(G)$ is stable by inverse and factors
		\item $L(G)$ is, up to the path-word correspondence, the set of decorated paths that admit a homomorphism to $G$: $L(G) = \set{w \in \Sigma^* \mid \text{path}(w) \ra G} $
		\item $L(G) \cup L(H) = L(G \cup H)$
		\item $L(G) \cap L(H) = L(G \times H)$ where $\times$ is the tensor (also called categorical) product 
		\item If $H$ is a cover of $G$, then $L(G)=L(H)$ (see \Cref{def:double-cover})
	\end{enumerate}
\end{proposition}
\footnotetext{Solution to \Cref{fig:puzzle}: \tiny{E2;E3;D3;C3;B3;B2;B1;C1}}

\begin{definition}\label{def:double-cover}
	$H$ is a cover of a graph $G$ if there is a covering map from $H$ to $G$.
	 A covering map $f$ is a surjection and a local isomorphism: the neighbourhood of a vertex $v$ in $H$ is mapped bijectively onto the neighbourhood of $f(v)$ in $G$.
\end{definition}

\begin{proof}[proof of the proposition]
	Let $G$, $H$ be decorated graphs, based on which we construct the automata $A_G$ and $A_H$.
	\begin{enumerate}
		\item If $w \in L(G)$, then $w$ is accepted by a path $q_0 \stack {a_1} \ra  q_1  \stack {a_2} \ra \dots \stack {a_k} \ra q_k$, by construction of $A_G$, $q_{i-1}  \stack {a_i} \ra q_{i} \iff q_{i} \stack {a_i^{-1}} \ra q_{i-1}$, therefore, $q_k \stack {a_k^{-1}} \ra \dots \stack {a_1^{-1}} \ra q_1$ is a valid path in $A_G$ which accepts $w^{-1}$. 
		
		That $L(G)$ is stable by factors is a consequence of every state of $A_G$ being initial and terminal.
		
		\item Consider a decorated path $P = p_0\dots p_k$ that admits a mapping $\phi$ to $G$, this mapping induces a walk $u_0\dots u_k$ in $G$ with the orientations and colors in $G$ matching that of $P$. In $A_G$, there must exist a run $u_0 \stack {a_1} \ra  u_1  \stack {a_2} \ra \dots \stack {a_k} \ra u_k$, with 
		$a_i = \begin{cases}
		\tra[c] \text{ if } p_{i-1}  \stack {c} \ra p_i \\
		 \tla[c] \text{ if } p_{i-1}  \stack {c} \la p_i
		 \end{cases}$
		 which accepts a word describing $P$.
		 Since $L(G)$ is stable by inverse, both words describing $P$ are in $L(G)$. 
			 With the same construction in reverse, it is immediate that $L(G)$ corresponds exactly to the paths mapping to $G$.
		 
		 \item Recall that for a connected graph $P$, $P \ra G \cup H \iff (P \ra G \text{ or } P \ra H )$, then, using the item $2$, we get the desired result.
		 \item Similarly, for a graph $P$, $P \ra G \times H \iff (P \ra G \text{ and } P \ra H )$, then, using the item $2$, we also get the desired result.
		 
		 \item Suppose $H$ is a cover of $G$ and let $p$ be the covering map from $H$ to $G$. Then $\set{(v,p(v) | v \in V(H)}$ is a bisimulation. Therefore, $G$ and $H$ have the same language.
	\end{enumerate}\end{proof}
	
	Since $A_G$ and $G$ share the same vertex set and accept the same paths/words, we will freely switch between graph-theoretic and automata-theoretic terminology. Thus, a path mapping to a graph may also be viewed as an accepting run over the graph. Similarly, we may view word languages as sets of paths. 
	\vspace{10pt}
	
	We say that an oriented path or cycle is \emph{directed} if all edges are oriented in the same direction (meaning clockwise or anticlockwise for a cycle).
	\begin{proposition}\label{prop:directed-cycle-necessary}
		If $G$ is an uncoloured oriented graph, then $L(G) = \Sigma^*$ if and only if $G$ contains a directed cycle.
		
		If $G$ is a 2-edge-colored undirected graph, then $L(G) = \Sigma^*$ if and only if $G$ contains a color-alternating cycle.
	\end{proposition}
	\begin{proof}
		Consider $G$ an uncoloured oriented graph, if $G$ contains a directed cycle $C$, every vertex of $C$ has in-degree and out-degree $1$ in $C$, so every letter ($\tra$ or $\tla$) can be accepted from every vertex and $L(G) = \Sigma^*$.
		
		Conversely, $G$ must accept directed paths $P$ of any length, since $G$ is finite, there exists two vertices $p_i,p_j \in V(P)$ that are mapped to the same vertex in $G$. And therefore, $P[p_i,\dots,p_j]$ maps to a directed closed walk of $G$. By taking the minimal one, we get a directed cycle.
		
		For $2$-edge-colored undirected graphs, the argument is the same, replacing \emph{directed}, by \emph{color-alternating}.
	\end{proof}
\begin{corollary}
	If $G$ is an uncoloured oriented graph or a 2-edge-colored undirected graph, $L(G) = \Sigma^*$ implies that there is a vertex $u$ such that $G$ accepts $\Sigma^*$ starting at $u$
\end{corollary}	

	Note that this result does not generalize to 3-edge-colored graphs, as shown in \Cref{fig:rainbow_triangle}.
	
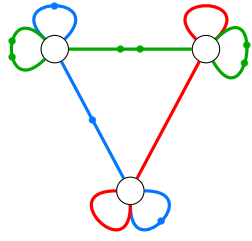
\begin{figure}[ht]
	\centering
		\begin{tikzpicture}
	\begin{pgfonlayer}{nodelayer}
		\node [style=black empty] (12) at (-2, 3.75) {};
		\node [style=black empty] (13) at (2, 3.75) {};
		\node [style=black empty] (14) at (0, 0) {};
	\end{pgfonlayer}
	\begin{pgfonlayer}{edgelayer}
		\draw [style={blue_edge}] (12) to (14);
		\draw [style={blue_edge}, in=135, out=45, loop] (12) to ();
		\draw [style={blue_edge}, in=-90, out=0, loop] (14) to ();
		\draw [style=edge red] (14) to (13);
		\draw [style=edge red, in=135, out=45, loop] (13) to ();
		\draw [style=edge red, in=-180, out=-90, loop] (14) to ();
		\draw [style=green edge] (12) to (13);
		\draw [style=green edge, in=-60, out=45, loop] (13) to ();
		\draw [style=green edge, in=135, out=-135, loop] (12) to ();
	\end{pgfonlayer}
\end{tikzpicture}
		\caption{A graph whose language is $\Sigma^*$ }
		\label{fig:rainbow_triangle}
\end{figure}
	
\begin{proposition}
	For any two graph-based languages $L_1,L_2$ and any connected decorated graph $G$, $L_1 \cup L_2 = L(G) \Ra L_1 = L(G)$ or $L_2 = L(G)$.
\end{proposition}
\begin{proof}
	Suppose $L_1 \cup L_2 = L(G)$ but $L_1,L_2 \subsetneq L(G)$. Let $w_1 \in L(G)\setminus L_1$ and $w_2 \in L(G)\setminus L_2$. 
	Let $u_1 \xrightarrow{w_1} v_1$ and $u_2 \xrightarrow{w_2} v_2$ be two runs in $G$ accepting $w_1$ and $w_2$ respectively. Since $G$ is connected, there exists a path from $v_1$ to $v_2$, call the corresponding word $w_3$. Then $w_1w_3w_2 \in L(G)$. 
	Since $L_1$ and $L_2$ are stable by factors, $w_1w_3w_2 \notin L_1$ and  $w_1w_3w_2 \notin L_2$ thus $w_1w_3w_2 \in L(G) \setminus(L_1 \cup L_2)$, so $L_1 \cup L_2 \neq L(G)$.
\end{proof}

From the preceding propositions, the reader might be tempted to think that the class of graph-based languages is trivial or not as rich as regular languages. We therefore show that, up to a simple encoding, any regular language can be expressed as a graph-based language. This construction is also useful for translating noteworthy automata into graphs.

\begin{proposition}
	Any regular language $L$ over an alphabet $\Sigma$ can be obtained as $$\tau(L(G) \cap \$\!\set{a,b}^*\!\#). $$
	Where $G$ is an edge colored graph with colors $\Omega = \set{\$,a,b,\#}$ and $\tau : \Omega^* \ra \Sigma^*$ is a transducer independent of $G$.
\end{proposition}
\begin{proof}
	Let $k = |\Sigma|$, for $\alpha$ the $i$-th letter of $\Sigma$, let $\phi(\alpha) = (ab)^{i-1}abbb(ab)^{n-i} $. 
	Then, let $\tau$ be a transducer that writes $\alpha$ when it reads $\phi(\alpha)$ and writes $\epsilon$ after reading $\$$ or $\#$. 	
	Now, we define $G$ from an automaton $A=(Q,\Omega,\delta,I,F)$ such that $L(A)=L$:\\
	We define $V(G) = Q \sqcup \textit{Aux}$, with \textit{Aux} a set of auxiliary vertices.\\
	For $q\in I$, we add a loop $q \arc\$ q $,\\
	for $q\in F$, we add a loop $q \arc \# q $,\\
	for $\delta(q,c) = q'$, we add a path $q \smash{\underbrace{\arc a x_1\arc b x_2\dots \arc b }_{\phi(c)}} q' $, with the $x_i$s being auxiliary vertices used for that transition only. 
	
	Now, it remains to verify that $\tau(L(G) \cap \$\!\set{a,b}^*\!\#)$ corresponds indeed to $L$.
	The fact that $L\subset \tau(L(G) \cap \$\!\set{a,b}^*\!\#)$ is relatively immediate, because, the definition of $G$ is such that the run in $A$ will naturally translate to a run in $G$ between $\$$ and $\#$, where each transition $c$ is replaced with $\phi(c)$.
	
	For the other direction, we need to prove that any word in $L(G)$ between $\$$ and $\#$ corresponds to $\phi(w)$ for some $w \in L$. In fact, we prove that, starting from some state $q\in Q \subset V(G)$, reading $\phi(c)$ can end at $q'$ if and only if $\delta(q,c)=q'$. 
	Call $q,x_1,\dots,q'$ a possible walk when reading $\phi(c)$.
	Since $\phi(c)$ starts with $a$, $x_1$ must be at the beginning of a path corresponding to some $\delta(q,c')$ (the path cannot be followed "in reverse"). In this path, most vertices have in-degree and out-degree 1 in $a$ and $b$, and therefore the automata acts deterministically there: the only possibility is to go forward in the path. 
	Then, on the part $abbba$, we must read exactly $abbba$ (not $aba$), thus for the run to be acceptable, we must have $c'=c$ and in that case $q' = \delta(q,c)$.  
\end{proof}

We also give another construction that doesn't need a transducer.

\begin{proposition}
	For $L$ a regular language over the alphabet $\Sigma$ and $k>0$ there exists a $\Sigma\cup \set{\$,\#}$-edge-colored graph $G$ such that $\set{w \mid \$w\# \in L(G) \text{ and } |w| = k} = \set{w \in L \mid |w|=k}$. 
\end{proposition}

\begin{proof}
	Let $A$ be an automaton recognising $L$. Take $G_1 = A \times P_k$, where $P_k$ is the automaton corresponding to a path of length $k$, with its initial and final states at the two endpoints, and accepting $\Sigma$ between two adjacent states.
	
	Then create $G$ from $G_1$ by adding a $\$$-loop to the initial states and a $\#$-loop to the final states.
\end{proof}

\subsection{Generalisation to trees}

In this section, we introduce the notion of \emph{graph-based tree-language}. This is particularly useful for graph theory, because the trees of a graph are, arguably, more important and meaningful than its paths.

The standard notions of tree automata (top-down automata, bottom-up automata, hedge automata) are not well suited for describing decorated trees of arbitrary arity with unordered children. Therefore, we do not describe an explicit tree automaton, instead, we directly discuss regular languages of decorated trees. We will, however, do so in the bounded-degree case. We follow the definitions and basic properties of \cite{EPTT17}.  %%Regular families of forests, antichains and duality pairs of relational structures

Fix a set $C$ of colors, we call $\mathbb F$ the set of oriented $C$-edge-colored trees.
For $T$ and $S$ in $\bb F$, let $t \in V(T), s \in V(S)$. Let $T\stackrel {t,s} \Join S$ be the tree obtained by joining $S$ and $T$ and identifying $s$ and $t$, the new vertex will typically be called $st$. If $s=t$, we write $S\stackrel t \Join T$.

If $T$ is a tree and $t$ is a neighbour of $s$, let $T_{s \angle t}$ be the subtree containing $s$ (as its root), $t$, and all descendants of $t$. (See \Cref{fig:tree_operation_example}.)

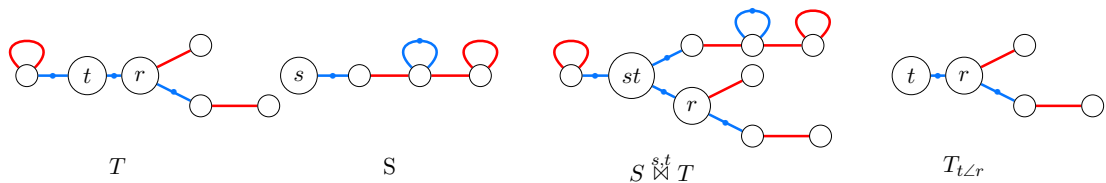
\begin{figure}[ht]
	\centering
	\scalebox{.8}{\begin{tikzpicture}
	\begin{pgfonlayer}{nodelayer}
		\node [style=black empty] (0) at (-5, 4) {$t$};
		\node [style=black empty] (1) at (-3.25, 4) {$r$};
		\node [style=black empty] (2) at (-1.25, 5) {};
		\node [style=black empty] (3) at (-1.25, 3) {};
		\node [style=black empty] (4) at (1, 3) {};
		\node [style=black empty] (5) at (-7, 4) {};
		\node [style=black empty] (6) at (8, 4) {};
		\node [style=black empty] (7) at (6, 4) {};
		\node [style=black empty] (8) at (4, 4) {};
		\node [style=black empty] (9) at (2, 4) {$s$};
		\node [style=none] (10) at (-4, 1) {$T$};
		\node [style=none] (11) at (5, 1) {S};
		\node [style=black empty] (13) at (15, 3) {$r$};
		\node [style=black empty] (14) at (17, 4) {};
		\node [style=black empty] (15) at (17, 2) {};
		\node [style=black empty] (16) at (19.25, 2) {};
		\node [style=black empty] (17) at (11, 4) {};
		\node [style=black empty] (18) at (19, 5) {};
		\node [style=black empty] (19) at (17, 5) {};
		\node [style=black empty] (20) at (15, 5) {};
		\node [style=black empty] (21) at (13, 4) {$st$};
		\node [style=none] (22) at (14, 1) {$S\stackrel {s,t} \Join T$};
		\node [style=none] (23) at (24, 1) {$T_{t \angle r}$};
		\node [style=black empty] (24) at (22.25, 4) {$t$};
		\node [style=black empty] (25) at (24, 4) {$r$};
		\node [style=black empty] (26) at (26, 5) {};
		\node [style=black empty] (27) at (26, 3) {};
		\node [style=black empty] (28) at (28.25, 3) {};
	\end{pgfonlayer}
	\begin{pgfonlayer}{edgelayer}
		\draw [style={blue_edge}] (0) to (1);
		\draw [style=edge red] (1) to (2);
		\draw [style={blue_edge}] (1) to (3);
		\draw [style=edge red] (3) to (4);
		\draw [style={blue_edge}] (0) to (5);
		\draw [style=edge red, in=135, out=45, loop] (5) to ();
		\draw [style=edge red] (7) to (6);
		\draw [style=edge red] (7) to (8);
		\draw [style={blue_edge}, in=135, out=45, loop] (7) to ();
		\draw [style={blue_edge}] (8) to (9);
		\draw [style=edge red, in=135, out=45, loop] (6) to ();
		\draw [style=edge red] (13) to (14);
		\draw [style={blue_edge}] (13) to (15);
		\draw [style=edge red] (15) to (16);
		\draw [style=edge red, in=135, out=45, loop] (17) to ();
		\draw [style=edge red] (19) to (18);
		\draw [style=edge red] (19) to (20);
		\draw [style={blue_edge}, in=135, out=45, loop] (19) to ();
		\draw [style={blue_edge}] (20) to (21);
		\draw [style=edge red, in=135, out=45, loop] (18) to ();
		\draw [style={blue_edge}] (21) to (13);
		\draw [style={blue_edge}] (21) to (17);
		\draw [style={blue_edge}] (24) to (25);
		\draw [style=edge red] (25) to (26);
		\draw [style={blue_edge}] (25) to (27);
		\draw [style=edge red] (27) to (28);
	\end{pgfonlayer}
\end{tikzpicture}}
	\caption{Illustration of the operations on trees}
	\label{fig:tree_operation_example}
\end{figure}

Given $L$ a family of trees, let $L / (T,t) = \set{S \mid \exists s \in V(S), T\stackrel {t,s} \Join S \in L}$. For $L$ a family of trees (possibly $L = \bb F$), we write $L^\star$ for the corresponding rooted trees, that is : $L^\star = \set{(T,t)| t \in V(T), T \in L}$.
We write $(T,t) \sim_L (S,s)$ if $L /  (T,t) = L / (S,s)$. We say that a set of decorated trees $L$ is \emph{regular} if the number of equivalence classes of $\sim_L$ is finite.  

The following property is folklore, it can be found in a similar language in \cite[Corollary 3.4]{EPTT17}
\begin{proposition}\label{prop:tree-mapping-regular}
	Given a decorated graph $G$, the set of decorated trees mapping to $G$ is regular.
\end{proposition}
We write this set $L^\tree(G)$.
\begin{proof}
Let $\Lambda_{(T,t) \ra G}$ be the different possible images of $t$ via a morphism from $T$ to $G$.
Then $T \in L^\tree(G) \iff \Lambda_{(T,t) \ra G} \neq \emptyset$ for $t \in V(T)$.
Furthermore, $T\stackrel {t,s} \Join S \in L^\tree(G) \iff \Lambda_{(T,t) \ra G} \cap \Lambda_{(S,s) \ra G} \neq \emptyset$. Thus the equivalence class of $\sim_{L^\tree(G)}$ for $(T,t)$ is only a function of $\Lambda_{(T,t) \ra G}$ and the number of equivalence classes is at most $2^{|G|}$ and $L^\tree(G)$ is regular.
\end{proof}

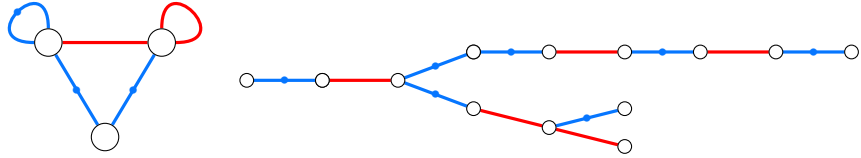
\begin{figure}[ht]
	\centering
	\begin{tikzpicture}
	\begin{pgfonlayer}{nodelayer}
		\node [style=black empty] (0) at (-10, 4) {};
		\node [style=black empty] (1) at (-7, 4) {};
		\node [style=black empty] (2) at (-8.5, 1.5) {};
		\node [style=small black empty] (3) at (-4.75, 3) {};
		\node [style=small black empty] (4) at (-2.75, 3) {};
		\node [style=small black empty] (5) at (-0.75, 3) {};
		\node [style=small black empty] (6) at (1.25, 2.25) {};
		\node [style=small black empty] (7) at (3.25, 1.75) {};
		\node [style=small black empty] (8) at (5.25, 2.25) {};
		\node [style=small black empty] (9) at (-2.75, 3) {};
		\node [style=small black empty] (10) at (1.25, 3.75) {};
		\node [style=small black empty] (11) at (3.25, 3.75) {};
		\node [style=small black empty] (12) at (5.25, 3.75) {};
		\node [style=small black empty] (13) at (7.25, 3.75) {};
		\node [style=small black empty] (14) at (1.25, 3.75) {};
		\node [style=small black empty] (15) at (9.25, 3.75) {};
		\node [style=small black empty] (16) at (11.25, 3.75) {};
		\node [style=small black empty] (17) at (5.25, 1.25) {};
	\end{pgfonlayer}
	\begin{pgfonlayer}{edgelayer}
		\draw [style=edge red] (0) to (1);
		\draw [style={blue_edge}] (0) to (2);
		\draw [style={blue_edge}] (2) to (1);
		\draw [style={blue_edge}, in=180, out=90, loop] (0) to ();
		\draw [style=edge red, in=90, out=0, loop] (1) to ();
		\draw [style={blue_edge}] (3) to (9);
		\draw [style={blue_edge}] (5) to (6);
		\draw [style={blue_edge}] (7) to (8);
		\draw [style=edge red] (9) to (5);
		\draw [style=edge red] (6) to (7);
		\draw [style={blue_edge}] (5) to (14);
		\draw [style={blue_edge}] (14) to (11);
		\draw [style=edge red] (11) to (12);
		\draw [style={blue_edge}] (12) to (13);
		\draw [style={blue_edge}] (15) to (16);
		\draw [style=edge red] (13) to (15);
		\draw [style=edge red] (7) to (17);
	\end{pgfonlayer}
\end{tikzpicture}
	\caption{A decorated graph and a tree from its language }
	\label{fig:ex_graph_tree_automata}
\end{figure}

\Cref{prop:basic-prop} generalises to tree.
\begin{proposition}\label{prop:basic-prop-tree}
	If $G$, $H$ are decorated graphs, then the following are true:
	\begin{enumerate}
		\item $L^\tree(G)$ is stable by subtree
		\item $L^\tree(G) \cup L^\tree(H) = L^\tree(G \cup H)$
		\item $L^\tree(G) \cap L^\tree(H) = L^\tree(G \times H)$ 		\item If $H$ is a cover of $G$, then $L^\tree(G)=L^\tree(H)$
	\end{enumerate}
\end{proposition}
\begin{proof}
	All of those are proved in the same way as in the proof of \Cref{prop:basic-prop}, as consequences of \Cref{prop:tree-mapping-regular}
\end{proof}

\begin{proposition}\label{prop:regular-graph-necessary}
	If $G$ is a decorated graph, then $L^\tree(G) = \bb F$ if and only if $G$ contains a subgraph where every vertex has in and out degree at least 1 in every color.
\end{proposition}
\begin{proof}
	If $H$ is a subgraph of $G$ with in and out degree at least 1 in every color, then a mapping of a decorated tree into $G$ can be found by mapping from the top down, choosing only from edges inside $H$.
	
	Conversely, if $L^\tree(G) = \bb F$. Let $T_k$ be the decorated tree of depth $k$ where every internal node has in and out degree at least 1 in every color. We call $r$ the central node of $T_k$, which we view as a root. $T_k$ must map to $G$ via a morphism $\phi$. If $k$ is large enough, then, in every path from the root to a leaf, there must be two vertices mapping to the same vertex of $G$. Consider $T_k'$ the subtree obtained by trimming $T_k$ is order that every leaf $v$ of $T_k'$ admits an ancestor $u$ such that $\phi(u) = \phi(v)$.
	
	Then, let $H = \phi(T)$, $H$ must have in and out degree at least 1 in every color as this is the case for every internal node of $T_k'$ and that the leaves have their image equal to that of a non-leaf.
\end{proof}

\Cref{prop:basic-prop-tree}, contrasting with \Cref{prop:directed-cycle-necessary} which only applies if $|\Sigma|=2$, illustrates how graph-based tree-automata can have tamer properties than graph-based automata, this will also be seen in \Cref{thm:minimal-tree-automata}.

In a context when we have trees with bounded degree, we can use more standard and constructive notions for regular sets, we will use the following :

\begin{definition}[Bottom-up decorated tree automaton]
Let $Q$ be a finite set of states and $F\subseteq Q$ a set of accepting states. A
\emph{bottom-up tree automata} is a tuple $A = (Q,\Gamma,F)$ where $\Gamma$ is a finite set of transitions of the form
\[
S(q_1,\dots,q_n)\to p,
\]
where $p,q_1,\dots,q_n\in Q$, and $S(q_1,\dots,q_n)$ is an decorated rooted star
whose root is labelled by $p$ and whose leaves are labelled by $q_1,\dots,q_n$.

A \emph{run} of $\Gamma$ on an decorated rooted tree $T$ is a state assignment to
the vertices such that, for every vertex, the corresponding edge-colored star belongs
to $\Gamma$. Leaves are handled by transitions of the form
\[
S()\to p.
\]
We say that $T$ is \emph{accepted} if it admits a run in which the root is assigned a
state in $F$. The language of $A$, denoted by $L(A)$, is the set of accepted
trees.
\end{definition}
In this paper, we will consider trees with unordered children, this is why we used a star in the transition rather than a tuple of decorated edges.

Also, the transitions may equivalently be viewed as grammar rules, with
$S(q_1,\dots,q_n)\to p$ describing the local replacement of a vertex in state $p$
by an edge-colored star whose leaves are in states $q_1,\dots,q_n$.

\begin{lemma}
	\label{lem:min-outside-L}
	Given a graph-based tree-language $L$, the set of minimal trees (for inclusion) not in $L$, is a regular set of trees with bounded degree.
\end{lemma} 
\begin{proof}
	Let $G$ be a graph such that $L = L^\tree(G)$. Let $M$ be the set of minimal trees outside $L$.
	For a (non-empty) rooted tree $(T,t)$, we consider $M/(T,t)$. Let $(S,s)$ be another non-empty rooted tree. 
	Let $\Lambda_{(T,t) \ra G}$ be the possible images of $t$ via a morphism from $T$ to $G$, and $\Lambda_{(T-1,t) \ra G} = \set{\Lambda_{(T',t) \ra G} \mid T' = T \setminus x, x \neq t}$ (actually, it will be enough to only consider the $x$ to be leaves). Then $T\stackrel {t,s} \Join S \in M$ if and only if the following are true : 
	\begin{enumerate}[itemsep=0pt] 
	\item There exists no morphisms from $T\stackrel {t,s} \Join S$ to $G$: $\Lambda_{(T,t) \ra G} \cap \Lambda_{(S,s) \ra G} = \emptyset$ . 
	\item  There exists a morphism if we remove one vertex (not $st$) from $T\stackrel {t,s} \Join S$:\\ 
	$$\forall \lambda \in \Lambda_{(T-1,t) \ra G}, \ \ \lambda \cap \Lambda_{(S,s) \ra G} \neq \emptyset$$
	 and similarly $\forall \lambda \in \Lambda_{(S-1,s) \ra G}, \ \ \lambda \cap \Lambda_{(T,t) \ra G} \neq \emptyset$.
	\item There exists a morphism if we remove $st$ from $T\stackrel {t,s} \Join S$: This is actually a consequence of 2. because $st$ is not a leaf and thus a cut vertex of $T\stackrel {t,s} \Join S$.
	%\item $\Lambda_{(T\setminus t,t) \ra G} \cap \Lambda_{(S\setminus t,s) \ra G} \neq \emptyset$. This condition follows from 2. 
	\end{enumerate} 
	Thus $M/(T,t)$ is characterised uniquely by $\Lambda_{(T,t) \ra G}$ and $\Lambda_{(T-1,t) \ra G}$.
	Since there are at most $2^{|G|}$ possibilities for $\Lambda_{T \ra G}$ and $2^{2^{|G|}}$ for $\Lambda_{(T-1,t) \ra G}$, $M$ is regular.
	
	Furthermore, if $T \in M$ has a vertex $s$ of degree $d$, say $N(s)=\set{t_1 \dots t_d}$, then, for $T$ to be minimal, we must have $\bigcap_{j \neq i} \Lambda_{(T_{s \angle t_i}\, ,\, s) \ra G} \neq \emptyset$ but $\bigcap_{j} \Lambda_{(T_{s \angle t_i}\, ,\, s) \ra G} = \emptyset$. Basic set theory implies that this can only happen if the ground set (here $G$) has size at least $d$. Therefore the trees of $M$ have maximum degree at most $n$.
	
	Note that using the possibilities for $\Lambda_{(T,t) \ra G}$ and $\Lambda_{(T-1,t) \ra G}$ as states, we can construct an explicit automata for $M$.
\end{proof}

\section{Characterisation of graph-based automata}
	
	\subsection{Path languages}
	
	\begin{definition}
		Fix a regular language $L$.
		For $w \in L$, let us define $$[w]\da  = \set{u \mathrel| \forall x, xw \in L \Ra xu \in L }$$ 
		
		For $q$ a state of an automaton recognizing $L$, let $L_q$ be the language of words starting from $q$. Then let $[L_q]\da = \bigcup_{w \in L_q} [w]\da$. (We might write $L_q^A$ and $[L_q^A]\da$ if $A$ is not clear from the context.)
	\end{definition}
	By Myhill-Nerode, since $L$ is regular, we get that there is a finite set of possible values for the $[L_q]\da$ (even if $q$ is a state from an infinite automata).
	
	\begin{lemma}[adding epsilon transitions]\label{lem:epsilon}
		Let $A$ be a (possibly infinite) automata and $r$, $s$ be two states verifying $[L_r]\da = [L_s]\da$. Then, if $A'$ is obtained from $A$ by adding the $\epsilon$-transition $r \stack{\eps}{\ra} s$, we have $L(A) = L(A')$.
		Furthermore, $[L_q]\da$ is unchanged from $A$ to $A'$.
	\end{lemma}
	\begin{proof}
		Since $A$ is a sub-automata of $A'$, $L(A) \subset L (A')$, so we just need to prove $L(A') \subset L(A)$.
		We do so by induction on the (minimum) number of times the new $\epsilon$-transition is used to accept a word.
		
		Let $w \in L(A')$.
		If $w$ is accepted without using the $\epsilon$-transition, then $w\in L(A)$.
		Otherwise suppose $w = uv$ is accepted by the path $q_0 \stack u {\ra}_{A'} r \stack \epsilon \ra s \stack v {\ra}_A q_1$. With $q_0 \stack u {\ra}_{A'} r$ using a minimal number of $\epsilon$-transitions.
		Observe that, in $A$,  $v \in L_s \subset [L_s]\da \subset [L_r]\da$, therefore, there exists $v' \in L_r$ such that $\forall x, xv' \in L(A) \Ra xv \in L(A)$. In particular $uv' \in L(A) \Ra uv \in L(A)$.
		Since $v'\in L_r$, we have $r \stack {v'} {\ra}_A q_1'$, thus $q_0 \stack u {\ra}_{A'} r \stack {v'} {\ra}_A q_1'$. But observe that this new path uses one less $\epsilon$-transition, so by induction $uv' \in L(A)$ which implies that $uv=w \in L(A)$.
		
		Furthermore, with the same reasoning, starting with $q_0 = r$, we get that $L_r^{A'} \subset [L_r]\da^A$. %Indeed, for a word $w=uv \in L_r^{A'}$ : $r \stack u {\ra}_{A'} r \stack \epsilon \ra s \stack v {\ra}_A q_1$, we will find $u'v' \in L_r^A$ with $u',v' \in L_r^A$ and $u \in [L_{u'}]\da, v \in [L_{v'}]\da$. Thus $w=uv \in [L_r]\da$ and $L_r^{A'} \subset [L_r^A]\da$.
		Then, by monotonicity of $[\bullet]\da$, $L_r^A \subset L_r^{A'}\subset [L_r^A]\da \Ra [L_r^A]\da \subset [L_r^{A'}]\da \subset [L_r^A]\da $.
		
	\end{proof}

	\begin{lemma}[quotient of infinite automata]\label{lem:infinite_automata}
		If $A_\infty$ is a possibly infinite automata accepting a regular language, then there exists $A$ a finite automata with the same language obtained by identifying states from $A_\infty$.
	\end{lemma}
	
	\begin{proof} Fix a language $L$ and let $A$ be obtained from $A_\infty$ by identifying all states with the same value for $[L_q]\da$.
		Using Lemma \ref{lem:epsilon}, we know we can add (bi-directional) epsilon transitions between all of those states, which is equivalent to merging them. Note that this is well defined because adding the $\eps$-transitions does not change the values of $[L_q]\da$.

	\end{proof}
	
	\begin{theorem}[characterisation of graph-based languages]
	\label{thm:language-characterisation}
		$L$ is the language of a graph-based automaton if and only if: 
		\begin{itemize}
			\item $L$ is a regular language
			\item $L$ is stable by factor : $ xuy \in L \Ra u \in L$
			\item $L$ verifies the "back and forth" property : $ xuy \in L \Ra xu u^{-1} u y \in L$
		\end{itemize}
	\end{theorem}
	
	\begin{proof}
		\boxed \Ra\ If $A$ is a graph-based automata accepting $L$, suppose $q_0 \stackrel{x}\ra q_1 \stackrel{u}\ra q_2 \stackrel{y}\ra q_3$, then $ q_1 \stackrel{u}\ra q_2$ and $q_0 \stackrel{x}\ra q_1 \stackrel{u}\ra q_2  \stackrel{u^{-1}}\ra  q_1 \stackrel{u}\ra q_2 \stackrel{y}\ra q_3$ are valid paths and therefore the corresponding words are in $L$.
		
		\boxed \La\ Fix a regular language $L$ satisfying the properties of the theorem. For $w\in L$, let $A_w$ be the graph-based automaton in the form of a path whose labels correspond to the letters of $w$. In other words, $A_w = q_0 \stackrel{w_0}{\lra} q_1 \stackrel{w_1}{\lra} \dots \stackrel{w_{|w|}}{\lra} q_{|w|}$.
		Observe that $w \in L(A_w)$. 
		First, we prove that every word in $L(A)$ can be obtained using the subword and back and forth operation (possibly multiple times) starting with $w$.
		Let $u \in L(A_w)$, and let $q_{i_0} \stackrel{u}\ra q_{i_m}$ be the path accepting $u$ with $m = |u|$. This path can be extended to $q_0 \stackrel{x}\ra q_{i_0} \stackrel{u}\ra q_{i_m} \stackrel{y}\ra q_{|w|} $. Let $v=xuy$. Observe that $u$ is a factor of $v$. We now show that $v$ can be obtained from $w$ by iteration of the \emph{back and forth} property. We do it by induction on the size of $v$.
		\begin{itemize}
			\item If $|v| \leq |w|$ : $v$ is accepted by a path starting at $q_0$ and ending at $q_w$. Since these two states are at distance $|w|$ in $A_w$, we conclude $v=w$.
			\item If $|v| > |w|$ : We call a path on $A_w$ a \emph{segment} if all transitions go in the same direction, meaning going all toward $q_{|w|}$ (right) or all away from it (left). 
			Note that a right-going segment from $q_i$ to $q_j$ accepts $w[i..j]$ while a left-going one accepts $w[i..j]^{-1}$.
			
			Decompose the path accepting $v$ into its (maximal) segments : $S_0,\dots,S_k$, accepting the words $s_0,\dots, s_k$ note that $S_1$ and $S_k$ must be going right and two consecutive segments must go in opposite directions (otherwise they could be merged in a bigger one). Since the path cannot go left of $q_0$ nor right of $q_{|w|}$, we have that $|S_0| \geq |S_1|$ and $|S_{k-1}| \leq |S_k|$.
			Let $S_i$ be a segment of minimal length, by the previous observation, we can choose $0<i<k$. Observe that $s_{i-1} s_i s_{i+1}$ contains $s_i^{-1}s_i{s_i}^{-1}$ as a factor. Let $v'$ be obtained from $v$ by replacing that occurrence of ${s_i}^{-1}s_i{s_i}^{-1}$ by ${s_i}^{-1}$. $v'$ is still in $A_w$ so by induction $v'$ is obtained from $w$ by iteration of the back and forth operation, and $v$ is obtained from $v'$ by a back and forth operation, which yields the desired result.

		\end{itemize}
		
		Therefore $L(A_w) \subset L$
		
		Now let $A_\infty = \bigcup_{w \in L} A_w$. We must have that $L(A_\infty) = L$. Note that $A_\infty$ is a graph-based automata and is, in general, infinite.
		
		Using Lemma \ref{lem:infinite_automata}, one deduces the existence of a finite automaton $A$ such that $L(A) = L(A_\infty) = L$, and $A$ is obtained from $A_\infty$ by identifying states. Therefore, $A$ is also a graph-based automaton.
	\end{proof}
	
	The above proof gives a way to construct a canonical graph-based automaton $A$ from a graph-based language. But, as usual for nondeterministic automata, this graph is not necessarily minimal, and there is no unique minimal graph-based automaton in general. See \Cref{fig:counter-example-minimality} for an example of two minimal graphs accepting the same language. 
	
	\begin{figure}[ht]
		\centering
		\begin{tikzpicture}
	\begin{pgfonlayer}{nodelayer}
		\node [style=black empty] (0) at (-5, 0) {};
		\node [style=black empty] (1) at (-2.5, 2.5) {};
		\node [style=black empty] (2) at (-2.5, -2.5) {};
		\node [style=black empty] (5) at (0, 0) {};
		\node [style=none] (6) at (-4.5, 1.75) {a,b};
		\node [style=none] (7) at (-1, 1.5) {c};
		\node [style=none] (8) at (-1, -1.5) {d};
		\node [style=none] (9) at (-4.25, -1.5) {b};
		\node [style=none] (10) at (2, 0) {\#,c,d};
		\node [style=none] (11) at (-7, 0) {\$,a,b};
		\node [style=none] (26) at (8.75, 1.5) {a};
		\node [style=none] (28) at (12.5, -2) {c,d};
		\node [style=black empty] (32) at (8, 0) {};
		\node [style=black empty] (33) at (10.5, 2.5) {};
		\node [style=black empty] (34) at (10.5, -2.5) {};
		\node [style=black empty] (35) at (13, 0) {};
		\node [style=none] (37) at (12.25, 1.5) {c};
		\node [style=none] (39) at (8.75, -1.5) {b};
		\node [style=none] (40) at (15, 0) {\#,c,d};
		\node [style=none] (41) at (6, 0) {\$,a,b};
		\node [style=none] (42) at (-2.5, 3.75) {a,b,c,d};
		\node [style=none] (43) at (-2.5, -3.75) {a,b,c,d};
		\node [style=none] (44) at (10.5, 3.75) {a,b,c,d};
		\node [style=none] (45) at (10.5, -3.75) {a,b,c,d};
		\node [shape=rectangle, fill=white, fill opacity=0.8] (46a) at (-2.5, 0) {a,b,c,d};
		\node [style=none] (46) at (-2.5, 0) {a,b,c,d};
		\node [shape=rectangle, fill=white, fill opacity=0.8] (47a) at (10.5, 0) {a,b,c,d};
		\node [style=none] (47) at (10.5, 0) {a,b,c,d};
	\end{pgfonlayer}
	\begin{pgfonlayer}{edgelayer}
		\draw [style=Edge, in=135, out=-135, loop] (0) to ();
		\draw [style=Edge, in=-45, out=45, loop] (5) to ();
		\draw [style=Edge, thick, draw={rgb,255: red,0; green,168; blue,0}, bend right=15] (0) to (1);
		\draw [style=Edge, thick, draw={rgb,255: red,0; green,168; blue,0}] (0) to (2);
		\draw [style=Edge, thick, draw={rgb,255: red,180; green,0; blue,20}] (2) to (5);
		\draw [style=Edge, thick, orange] (1) to (5);
		\draw [style=Edge, thick, blue, bend left=15] (0) to (1);
		\draw [style=Edge, thick, blue, in=135, out=-180, loop] (1) to ();
		\draw [style=Edge, thick, draw={rgb,255: red,180; green,0; blue,20}, in=0, out=45, loop] (1) to ();
		\draw [style=Edge, thick, orange, in=45, out=90, loop] (1) to ();
		\draw [style=Edge, thick, draw={rgb,255: red,0; green,168; blue,0}, in=90, out=135, loop] (1) to ();
		\draw [style=Edge, thick, blue, in=-135, out=180, loop] (2) to ();
		\draw [style=Edge, thick, draw={rgb,255: red,180; green,0; blue,20}, in=0, out=-45, loop] (2) to ();
		\draw [style=Edge, thick, orange, in=-90, out=-45, loop] (2) to ();
		\draw [style=Edge, thick, draw={rgb,255: red,0; green,168; blue,0}, in=-135, out=-90, loop] (2) to ();
		\draw [style=Edge, thick, blue, bend right] (1) to (2);
		\draw [style=Edge, thick, draw={rgb,255: red,180; green,0; blue,20}, bend right] (2) to (1);
		\draw [style=Edge, thick, orange, bend left=15, looseness=0.75] (1) to (2);
		\draw [style=Edge, thick, draw={rgb,255: red,0; green,168; blue,0}, bend left=15, looseness=0.75] (2) to (1);
		\draw [style=Edge, in=135, out=-135, loop] (32) to ();
		\draw [style=Edge, in=-45, out=45, loop] (35) to ();
		\draw [style=Edge, thick, draw={rgb,255: red,0; green,168; blue,0}] (32) to (34);
		\draw [style=Edge, thick, draw={rgb,255: red,180; green,0; blue,20}, bend right=15] (34) to (35);
		\draw [style=Edge, thick, orange] (33) to (35);
		\draw [style=Edge, thick, draw=blue] (32) to (33);
		\draw [style=Edge, thick, blue, in=135, out=-180, loop] (33) to ();
		\draw [style=Edge, thick, draw={rgb,255: red,180; green,0; blue,20}, in=0, out=45, loop] (33) to ();
		\draw [style=Edge, thick, orange, in=45, out=90, loop] (33) to ();
		\draw [style=Edge, thick, draw={rgb,255: red,0; green,168; blue,0}, in=90, out=135, loop] (33) to ();
		\draw [style=Edge, thick, blue, in=-135, out=180, loop] (34) to ();
		\draw [style=Edge, thick, draw={rgb,255: red,180; green,0; blue,20}, in=0, out=-45, loop] (34) to ();
		\draw [style=Edge, thick, orange, in=-90, out=-45, loop] (34) to ();
		\draw [style=Edge, thick, draw={rgb,255: red,0; green,168; blue,0}, in=-135, out=-90, loop] (34) to ();
		\draw [style=Edge, thick, blue, bend right] (33) to (34);
		\draw [style=Edge, thick, draw={rgb,255: red,180; green,0; blue,20}, bend right] (34) to (33);
		\draw [style=Edge, thick, orange, bend left=15, looseness=0.75] (33) to (34);
		\draw [style=Edge, thick, draw={rgb,255: red,0; green,168; blue,0}, bend left=15, looseness=0.75] (34) to (33);
		\draw [style=Edge, thick, orange, bend left=15] (34) to (35);
		\draw [thick, orange, in=45, out=135, loop] (5) to ();
		\draw [thick, draw={rgb,255: red,180; green,0; blue,20}, in=-135, out=-45, loop] (0) to ();
		\draw [thick, blue, in=60, out=135, loop] (0) to ();
		\draw [thick, draw={rgb,255: red,0; green,168; blue,0}, in=-135, out=-45, loop] (5) to ();
		\draw [thick, orange, in=45, out=135, loop] (35) to ();
		\draw [thick, draw={rgb,255: red,180; green,0; blue,20}, in=-120, out=-45, loop] (35) to ();
		\draw [thick, blue, in=45, out=135, loop] (32) to ();
		\draw [thick, draw={rgb,255: red,0; green,168; blue,0}, in=-135, out=-45, loop] (32) to ();
	\end{pgfonlayer}
\end{tikzpicture}
		\caption{Two minimal graphs for the same language, colours indicate the labels $a,b,c,d$.}
		\label{fig:counter-example-minimality}
	\end{figure}
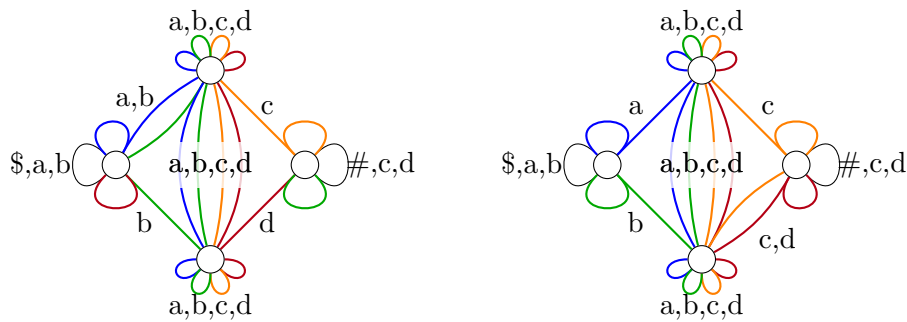
	
	The two graphs accept the same language $L$, which can be defined through forbidden factors:\\
	$L = \Sigma^* \setminus \left(\bigcup_{w \in F\cup F^{-1}} \Sigma^*w\Sigma^*\right)$ with $\Sigma = \set{\$,\#,a,b,c,d}, F = \set{ \$\#, \$c,\$d, a\#,b\#, \$ad\#}$.
	They are minimal in terms of size because a graph needs at least 3 different states: the states accepting $\$$ and $\#$ are non-adjacent. However, three states are not sufficient to accept $\$ac\#$ and $\$bd\#$ but not $\$ad\#$.

	\subsection{Tree languages}
	Now we give a characterisation similar to \Cref{lem:epsilon} for tree automata.
	
	\begin{lemma}\label{lem:infinite_automata_tree}
		If $G_\infty$ is an infinite graph accepting a regular tree language, then there exists $G$ a finite graph with the same tree-language obtained by identifying vertices from $G_\infty$.
	\end{lemma}
	\begin{proof}
		Fix a regular tree language $L$, for $(T,t) \in L^\star$, let $$[T,t]\da = \set{(S,s) | \forall (X,x) \in \bb F^\star, X\stackrel {x,t} \Join T \in L \Ra X\stackrel {x,s} \Join S \in L }$$
		For $u$ a vertex of a graph $G_\infty$ recognising $L$, let $L_u$ be the language of trees starting from $u$, then let $[L_u]\da = \bigcup_{T \in L_u} [T,u]\da$. 
		We claim that identifying all vertices of $G_\infty$ with the same value for $[L_u]\da$, will yield a finite graph $G$ with $L^\tree(G) = L^\tree(G_\infty)$. 
		
		In order to be able to describe the trees in $G$ based on the vertices in $G_\infty$, we add $\epsilon$-labeled edges between all pairs of identified vertices in $G_\infty$, we write $G_\infty^\epsilon$, we have $L^\tree(G) = L^\tree(G_\infty^\epsilon)$. Then, edges of $G$ lift to edges of $G_\infty^\epsilon$ and vertices of $G$ lift to $\epsilon$-labeled cliques of $G_\infty^\epsilon$. (see \Cref{fig:triangle_lift}).
		
		Since $G$ is obtained from identifying vertices of $G_\infty$, we have $L^\tree(G_\infty) \subset L^\tree(G) $. 
		Now, we prove $L^\tree(G) = L^\tree(G_\infty^\epsilon) \subset L^\tree(G_\infty)$.
		%We do an induction on the minimum number of $\epsilon$-edges needed in to accept a tree $T$ in $G_\infty^\epsilon$. 
		Consider some tree $T \in L^\tree(G_\infty^\epsilon)$, for the following analysis, in a slight abuse of notation, we will consider the $\epsilon$-edges used in the accepting run as edges of the tree itself (as in \Cref{fig:triangle_lift}, top right).
				
		If $T$ is accepted without using the $\epsilon$-edges, then $T\in L^\tree(G_\infty)$.
		Otherwise, let $T = R \stackrel {s} \Join S$ such that: 
		\begin{itemize} [itemsep=0pt]
			\item $s$ is $\eps$-adjacent to some vertex $u$.
			\item $S$ does not use any $\epsilon$-edges, 
			\item $R = T_{s \angle u}$, 
		\end{itemize}
		
		Since $u$ and $s$ are $\epsilon$-adjacent, we must have in $G_\infty$: $[L_{u}]\da = [L_{s}]\da \ni (S,s)$, therefore, there exists $(U,u) \in L_{u}$ such that:
		 $$ R\stackrel {u} \Join U \in L^\tree(G_\infty) \Ra \left( R\stackrel {s} \Join S = T \in L^\tree(G_\infty) \right).$$ 
		In $R\stackrel {u} \Join U$, the edge $su$, has to be a pendant edge.
		This edge is thus useless for the language of $R\stackrel {u} \Join U$ and we can therefore remove it.
		
		Repeating the procedure, we end up with a tree $T'$ of $G_\infty$ such that $T' \in L^\tree(G_\infty) \Ra T\in L^\tree(G_\infty)$. Thus $T \in L^\tree(G_\infty)$.
	\end{proof}
	
	\begin{figure}[ht]
		\centering
		\begin{tikzpicture}
	\begin{pgfonlayer}{nodelayer}
		\node [style=none] (18) at (-3.75, 7.75) {$\epsilon$};
		\node [style=black empty] (19) at (-4, 6) {2};
		\node [style=black empty] (21) at (-2.75, 4) {4};
		\node [style=black empty] (23) at (-1.5, 6) {3};
		\node [style=black empty] (24) at (-4, -0.5) {2};
		\node [style=black empty] (25) at (-1.5, -0.5) {13};
		\node [style=none] (27) at (-2.75, 3) {};
		\node [style=none] (28) at (-2.75, 0.25) {};
		\node [style=black empty] (29) at (-6.5, 6) {1};
		\node [style=black empty] (31) at (-2.75, -2.5) {4};
		\node [style=small black empty] (32) at (3, -1.5) {};
		\node [style=small black empty] (33) at (5, -1.5) {};
		\node [style=small black empty] (34) at (7, -1.5) {};
		\node [style=small black empty] (35) at (9, -1.5) {};
		\node [style=small black empty] (36) at (5, -1.5) {};
		\node [style=none] (37) at (2, -1.5) {};
		\node [style=none] (38) at (-1, -1.5) {};
		\node [style=small black empty] (39) at (5, 4.75) {};
		\node [style=small black empty] (41) at (9, 4.75) {};
		\node [style=small black empty] (42) at (11, 4.75) {};
		\node [style=small black empty] (43) at (7, 4.75) {};
		\node [style=small black empty] (44) at (3, 4.75) {};
		\node [style=none] (45) at (6, 5.5) {$\epsilon$};
		\node [style=none] (46) at (2.25, 4.75) {};
		\node [style=none] (47) at (-0.75, 4.75) {};
		\node [style=none] (48) at (5, -1) {13};
		\node [style=none] (49) at (7, -1) {2};
		\node [style=none] (50) at (3, -1) {2};
		\node [style=none] (51) at (9, -1) {4};
		\node [style=none] (52) at (5, 4.25) {1};
		\node [style=none] (53) at (9, 4.25) {2};
		\node [style=none] (54) at (3, 4.25) {2};
		\node [style=none] (55) at (11, 4.25) {4};
		\node [style=none] (56) at (7, 4.25) {3};
		\node [style=none] (57) at (6, 1) {};
		\node [style=none] (58) at (6, 3) {};
		\node [style=none] (59) at (-3.25, 1.75) {identifying vertices};
		\node [style=none] (60) at (6.75, 2) {lift};
		\node [style=none] (61) at (0.75, -1) {mapping};
		\node [style=none] (62) at (0.75, 5.25) {mapping};
		\node [style=small black empty] (63) at (8.25, 3.25) {};
		\node [style=none] (64) at (-6.5, 4.25) {$G_\infty^\epsilon$};
		\node [style=none] (65) at (-6.5, -1) {$G$};
		\node [style=none] (66) at (11, -1) {$T$};
		\node [style=none] (67) at (12.75, 4.25) {$T$};
		\node [style=none] (68) at (12.75, 3.5) {\tiny{(with epsilon transition)}};
		\node [style=small black empty] (69) at (6, -3) {};
	\end{pgfonlayer}
	\begin{pgfonlayer}{edgelayer}
		\draw [style={blue_edge}] (19) to (23);
		\draw [style=edge red] (23) to (21);
		\draw [style={blue_edge}] (24) to (25);
		\draw [style=Arc] (27.center) to (28.center);
		\draw [style={blue_edge}] (29) to (19);
		\draw [style={dashed_edge}, bend left=90, looseness=0.50] (29) to (23);
		\draw [style=edge red] (25) to (31);
		\draw [style=edge red] (21) to (19);
		\draw [style=edge red] (31) to (24);
		\draw [style={blue_edge}] (32) to (36);
		\draw [style={blue_edge}] (36) to (34);
		\draw [style=edge red] (34) to (35);
		\draw [style=Arc] (37.center) to (38.center);
		\draw [style={blue_edge}] (43) to (41);
		\draw [style=edge red] (41) to (42);
		\draw [style={blue_edge}] (44) to (39);
		\draw [style={dashed_edge}] (39) to (43);
		\draw [style=Arc] (46.center) to (47.center);
		\draw [style=Arc] (57.center) to (58.center);
		\draw [style=edge red] (43) to (63);
		\draw [style=edge red] (36) to (69);
	\end{pgfonlayer}
\end{tikzpicture}
		\caption{A lift of a path after identification of vertices}
		\label{fig:triangle_lift}
	\end{figure}
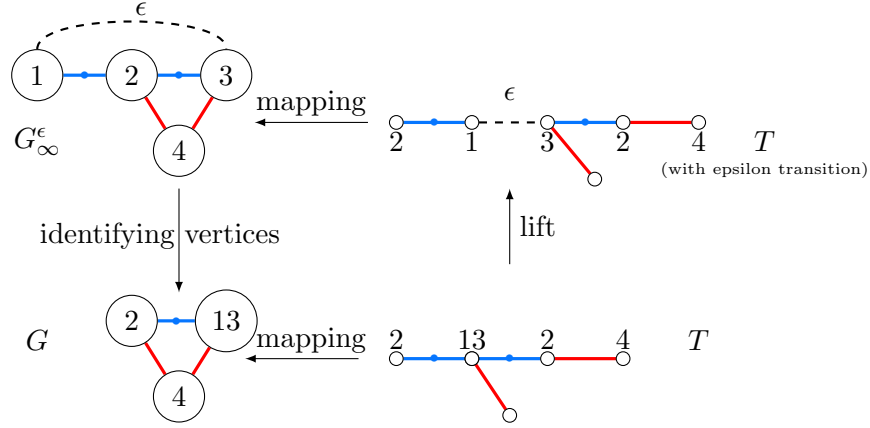
	
	\begin{figure}[ht]
		\centering
		\scalebox{1}{\begin{tikzpicture}
	\begin{pgfonlayer}{nodelayer}
		\node [style=none] (19) at (0, 2.25) {};
		\node [style=none] (20) at (-2.75, 7.5) {};
		\node [style=none] (21) at (2.75, 7.5) {};
		\node [style=none] (22) at (0, 9.25) {};
		\node [style=none] (4) at (0, 3.5) {};
		\node [style=none] (5) at (-1.75, 0.75) {};
		\node [style=none] (6) at (1.75, 0.75) {};
		\node [style=none] (7) at (0, -0.25) {};
		\node [style=black empty] (8) at (0, 5) {};
		\node [style=black empty] (9) at (0, 2.75) {};
		\node [style=black empty] (10) at (-1, 0.75) {};
		\node [style=black empty] (11) at (1, 0.75) {};
		\node [style=black empty] (12) at (0, 7) {};
		\node [style=black empty] (13) at (-1.25, 5) {};
		\node [style=black empty] (14) at (1.25, 5) {};
		\node [style=black empty] (15) at (1, 7) {};
		\node [style=black empty] (16) at (2, 7) {};
		\node [style=none] (17) at (0, 2.75) {$s$};
		\node [style=none] (18) at (0, 5) {$u$};
		\node [style=black empty] (23) at (-2, 7) {};
		\node [style=none] (24) at (0, 8) {$R$};
		\node [style=none] (25) at (0, 0.25) {$S$};
		\node [style=black empty] (26) at (-1, 7) {};
		\node [style=none] (27) at (8.75, 2.5) {};
		\node [style=none] (28) at (6, 7.5) {};
		\node [style=none] (29) at (11.5, 7.5) {};
		\node [style=none] (30) at (8.75, 9.25) {};
		\node [style=none] (31) at (8.75, 5.75) {};
		\node [style=none] (32) at (7, 2.75) {};
		\node [style=none] (33) at (10.5, 2.75) {};
		\node [style=none] (34) at (8.75, 1.75) {};
		\node [style=black empty] (35) at (8.75, 5) {};
		\node [style=black empty] (36) at (8.75, 5) {};
		\node [style=black empty] (37) at (8.75, 3) {};
		\node [style=black empty] (38) at (7.75, 3) {};
		\node [style=black empty] (39) at (8.75, 7) {};
		\node [style=black empty] (40) at (7.5, 5) {};
		\node [style=black empty] (41) at (10, 5) {};
		\node [style=black empty] (42) at (9.75, 7) {};
		\node [style=black empty] (43) at (10.75, 7) {};
		\node [style=none] (45) at (8.75, 5) {$u$};
		\node [style=black empty] (46) at (6.75, 7) {};
		\node [style=none] (47) at (8.75, 8) {$R$};
		\node [style=none] (48) at (9.5, 2.5) {$U$};
		\node [style=black empty] (49) at (7.75, 7) {};
		\node [style=none] (50) at (0.25, 4) {$\epsilon$};
		\node [style=none] (51) at (-2, 6) {$\epsilon$};
		\node [style=none] (52) at (6.75, 6) {$\epsilon$};
	\end{pgfonlayer}
	\begin{pgfonlayer}{edgelayer}
		\draw [style={shape_border}] (22.center)
			 to (21.center)
			 to [bend left, looseness=0.75] (19.center)
			 to [bend left, looseness=0.75] (20.center)
			 to cycle;
		\draw [style={shape_border}] (7.center)
			 to [bend right, looseness=1.25] (6.center)
			 to (4.center)
			 to (5.center)
			 to [bend right, looseness=1.25] cycle;
		\draw [style={blue_edge}] (9) to (10);
		\draw [style=edge red] (9) to (11);
		\draw [style={dashed_edge}] (8) to (9);
		\draw [style={blue_edge}] (12) to (8);
		\draw [style={blue_edge}] (14) to (15);
		\draw [style=edge red] (14) to (16);
		\draw [style={dashed_edge}] (13) to (23);
		\draw [style=edge red] (26) to (18.center);
		\draw [style=edge red] (13) to (18.center);
		\draw [style={dashed_edge}] (18.center) to (14);
		\draw [style={shape_border}] (30.center)
			 to (29.center)
			 to [bend left, looseness=0.75] (27.center)
			 to [bend left, looseness=0.75] (28.center)
			 to cycle;
		\draw [style={shape_border}] (34.center)
			 to [bend right, looseness=1.25] (33.center)
			 to (31.center)
			 to (32.center)
			 to [bend right, looseness=1.25] cycle;
		\draw [style={blue_edge}] (36) to (37);
		\draw [style=edge red] (36) to (38);
		\draw [style={blue_edge}] (39) to (35);
		\draw [style={blue_edge}] (41) to (42);
		\draw [style=edge red] (41) to (43);
		\draw [style={dashed_edge}] (40) to (46);
		\draw [style=edge red] (49) to (45.center);
		\draw [style=edge red] (40) to (45.center);
		\draw [style={dashed_edge}] (45.center) to (41);
	\end{pgfonlayer}
\end{tikzpicture}}
		\caption{Illustration of the proof of \Cref{lem:infinite_automata_tree}}
		\label{fig:remove_epsilon}
	\end{figure}
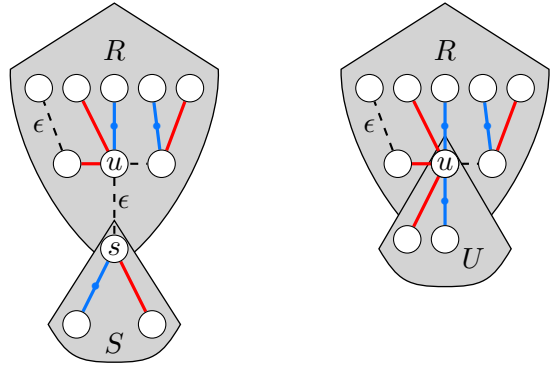
	
	\begin{theorem}[characterisation of graph-based tree-languages]
		\label{thm:tree-language-characterisation}
		$L$ is a graph-based tree-language if and only if: 
		\begin{itemize}
			\item $L$ is a regular tree-language
			\item $L$ is stable by subtree
			\item $L$ is stable via "doubling": if $T \in L$ and $t \in V(T)$, then $T\stackrel t \Join T \in L$.
		\end{itemize}
	\end{theorem}

	\begin{lemma}\label{lem:tree-language}
		For a tree $T$, $L^\tree(T)$ is exactly the language obtained from $T$ by repeating the operations of subtrees and doubling.
	\end{lemma}
	\begin{proof}
		First, observe that $T \in L^\tree(T)$ and, if $S \in L^\tree(T) $(meaning $S \ra T$), then for $S'$ a subtree of $S$, $S' \ra S \ra T$, furthermore, for $s \in V(S)$, $S\stackrel s \Join S \ra S \ra T$. Thus, all trees obtained from $T$ by the subtree and doubling operations are in $L^\tree(T)$.
		
		For the other direction, let $S \ra T$, we want to show that $S$ can be obtained using the doubling and subgraph operation. 
		We will start with the following observation : if $S_1$ and $S_2$ are subtrees of $T_1$ and $T_2$ respectively, via the embeddings $\phi_1$ and $\phi_2$, and $\phi_1(s_1) = t_1 \in T_1, \phi_2(s_2) = t_2 \in T_2$ for $s_1 \in V(S_1)$ and $s_2 \in V(S_2)$, then $S_1 \stackrel {s_1,s_2} \Join S_2 $ is a subtree of $ T_1 \stackrel {t_1,t_2} \Join T_2$ with $s_{1}s_{2}$ mapped to $t_{1}t_{2}$. ($*$)
		
		We actually show a stronger result: if $(S,s) \arc\phi (T,t)$, then there exists $T'$ obtain by repeated doublings of $T$ such that $(S,s)$ is a subgraph of $(T',t)$ (meaning that the subtree embedding maps $s$ to $t$).
			
	If $S$ is at most an edge, then $S$ is a subgraph of $T$ and $T'=T$ convene. Otherwise, let $s$ be a non-leaf of $S$ with neighbours $s_1,\dots,s_k$ and consider the subtrees $S_i = T_{s \angle s_i}$. By induction, each $(S_i,s)$ is a subgraph of some $(T_i,t)$ obtained via doublings of $T$.
	Then, using, ($*$), we get that $\Join_{i \leq k }\!\!(S_i,s)$ is a subtree of $\Join_{i \leq k }\!\!(T_i,t) $ which can be obtained via doublings of $T$ (and taking a subtree). 
	\end{proof}
	
	Now we can prove the \Cref{thm:tree-language-characterisation}, in a similar way to \Cref{thm:language-characterisation}
	
	\begin{proof}[proof of \Cref{thm:tree-language-characterisation}]
		\boxed \Ra\ If $G$ is a graph accepting $L$, and $T \in L$. In \Cref{lem:tree-language}, we saw that the set of trees accepted by $T$, is exactly the trees obtained by the subtrees and doubling operations. Since $S \in L^\tree(T)$ is equivalent to $S \ra T$, by transitivity of morphisms, we get that $L^\tree(G)$ is stable by subtree and doubling.
		
		\boxed \La\ Fix a regular language $L$ verifying the properties of the theorem. 		
		Now let $G_\infty = \bigcup_{T \in L} T$. For any $T$, by \Cref{lem:tree-language}, $L^\tree(T)$ is exactly the set of trees obtained by doubling and subtree starting from $T$, but since $L$ is stable by those operations, we must have that $L^\tree(G_\infty) = L$. Note that $G_\infty$ is a graph and is, in general, infinite.
		
		Using \Cref{lem:infinite_automata_tree}, one deduces the existence of a finite automaton $A$ such that $L^\tree(A) = L^\tree(A_\infty) = L$, and $A$ is obtained from $A_\infty$ by identifying states. Therefore, $A$ is also a graph-based tree automaton.
	\end{proof}

	\begin{lemma}[maximal simplices]\label{lem:max_simplices}
		If $L$ is a set and $\cal C$ is a set of finite subsets of $L$, which is downward closed, meaning $(A\! \in\! \cal C) \wedge (B\! \subset\! A) \Ra (B \in \cal C)$,  then there exists a set $\cal M$ verifying $\cal C = \bigcup_{M \in \cal M} \set{A \mid A \subset_{finite} M}$ with every element of $\cal M$ being maximal, meaning that adding anything to it will induce a finite subset outside $\cal C$.
		
		Furthermore, when $|\cal M|$ is finite, it is the unique minimum set in terms of cardinality such that $\cal C = \bigcup_{M \in \cal M} \set{A \mid A \subset_{finite} M}$.
	\end{lemma}
	
	\begin{proof}
		What we have defined is a simplicial complex, whose maximal elements are called maximal faces. In that context, this result is well-known. Note that, for arbitrary infinite sets, it follows from Zorn's Lemma.
		The second part of the proof follows from the observation that the maximal faces are uniquely defined and must be part of each finite solution.
	\end{proof}

	We have found a characterisation for graph-based languages and tree-languages. This proves, given a suitable language, the existence of a graph recognising it. We now show that, in the case of tree automata, there also exists a unique minimal one.
	
	\begin{theorem}\label{thm:minimal-tree-automata}
		For a graph-based tree-language $L$, there exists a unique minimal graph $G$ that accepts $L$ (with respect to the number of vertices).
	\end{theorem}
	\begin{proof}
		Let $L$ be a graph-based tree-language.
		Two rooted trees $(S,s)$ and $(T,t)$ are said to be \emph{compatible} if $(S \stackrel {s,t} \Join T,st) \in L^\star$. Note that, if $(S,s)$ and $(T,t)$ are accepted starting from the same vertex in a given graph recognising $L$, then $(S,s)$ and $(T,t)$ must be compatible. Similarly, we say that a finite set of rooted trees is compatible if the join of all of them is still in $L$. Let $\cal K = K_1,K_2,\dots$ be the maximal compatible sets obtained from $\Cref{lem:max_simplices}$. By maximality, the $K_i$'s are stable by joining over $t$ and taking subtrees (still containing $t$).
		
		Consider the graph $\G$ whose vertex set is $\cal K$ and $K \arc c K'$ if for every $(T,t) \in K$, the rooted tree $(T \stackrel {t} \Join (t \arc c t'),t')$ obtained by adding a new root $t'$ to $T$ is in $ K'$. 
		Note that, by compatibility, this is the same as saying that, for all $(T,t) \in K, (T',t') \in K'$ we have $(T \stackrel {t} \Join (t \arc c t') \stackrel {t'} \Join T',t') \in K'$ and the same tree rooted at $t$ in $K$. This presentation manifests the fact that $K \arc c K' \iff K' \xleftarrow c K$.
		
		We now claim two things : (1) $L^\tree(\G) = L$ and (2) Every graph accepting $L$ has more vertices than $\G$.
		\vspace{5pt}
		
		(1): We prove that the language of $\G$ starting at a vertex $K$ is at least $K$. We do so by induction on the depth of a rooted tree $(T,t)$. Trees of depth $0$ can map anywhere in $\G$. If $(T,t)$ is of depth $d$, let the $t_i$'s be the neighbours of $t$ and $T_i = T_{t \angle t_i}$. 
By definition of $\cal K$, there must exist some $K \in V(\G)$ such that $(T,t) \in K$. Then let $X_i = \set{(R \stackrel {t} \Join T_i,t_i) \mid (R,t) \in K}$, note that we have $(R \stackrel {t} \Join T_i,t) \in K$ by compatibility between $K$ and $T_i$.

We prove that $X_i$ is a compatible set. Given a finite family $((S_j,t_i)_{j \leq j_{max}})$ in $X_i$, with each $S_j = R_j \stackrel {t} \Join T_i$, then let $R_{all} = \Join_{j \leq j_{max}}\!\! (R_j,t)$, and $S_{all} = R_{all} \stackrel {t} \Join T_j$. Observe that $(S_{all},t)$ and $(R_{all},t)$ are in $K$. By doubling $S_{all}$ on $t_i$, repeating $j_{max}$ times, we obtain a tree containing $\Join_{j \leq j_{max}}\!\! (S_j,t_i)$ as a subtree, which therefore must be in $L$ (see the two rightmost trees in \Cref{fig:tree_unfolds}). This proves that $X_i$ is a compatible set.
By definition of $\cal K$, there must exists some $K_i \supset X_i$.

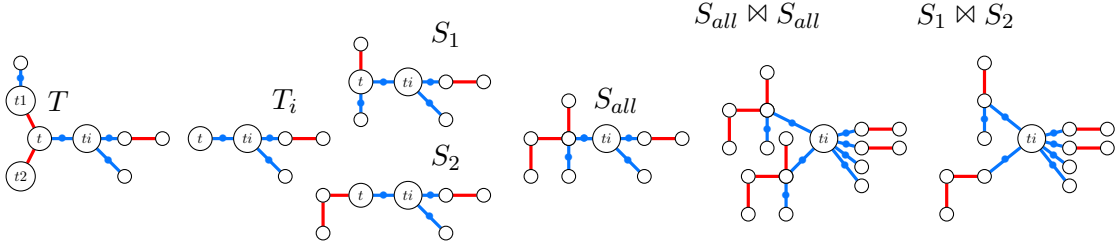
\begin{figure}[ht]
		\centering
		\begin{tikzpicture}
	\begin{pgfonlayer}{nodelayer}
		\node [style=small black empty] (0) at (-2, 0) {$ti$};
		\node [style=small black empty] (1) at (-1, 0) {};
		\node [style=small black empty] (2) at (0, 0) {};
		\node [style=small black empty] (3) at (-1, -1) {};
		\node [style=none] (4) at (-1, 1) {$T_i$};
		\node [style=small black empty] (5) at (2.25, -1.5) {$ti$};
		\node [style=small black empty] (6) at (3.25, -1.5) {};
		\node [style=small black empty] (7) at (4.25, -1.5) {};
		\node [style=small black empty] (8) at (3.25, -2.5) {};
		\node [style=none] (9) at (3.25, 2.75) {$S_1$};
		\node [style=small black empty] (10) at (1, -1.5) {$t$};
		\node [style=small black empty] (11) at (0, -1.5) {};
		\node [style=small black empty] (12) at (0, -2.5) {};
		\node [style=small black empty] (13) at (2.25, 1.5) {$ti$};
		\node [style=small black empty] (14) at (3.25, 1.5) {};
		\node [style=small black empty] (15) at (4.25, 1.5) {};
		\node [style=small black empty] (16) at (3.25, 0.5) {};
		\node [style=none] (17) at (3.25, -0.5) {$S_2$};
		\node [style=small black empty] (18) at (1, 1.5) {$t$};
		\node [style=small black empty] (19) at (1, 0.5) {};
		\node [style=small black empty] (21) at (1, 2.5) {};
		\node [style=small black empty] (27) at (6.5, 0) {};
		\node [style=small black empty] (28) at (5.5, 0) {};
		\node [style=small black empty] (29) at (5.5, -1) {};
		\node [style=small black empty] (30) at (7.5, 0) {$ti$};
		\node [style=small black empty] (31) at (8.5, 0) {};
		\node [style=small black empty] (32) at (9.5, 0) {};
		\node [style=small black empty] (33) at (8.5, -1) {};
		\node [style=none] (34) at (7.75, 1) {$S_{all}$};
		\node [style=small black empty] (35) at (6.5, 0) {};
		\node [style=small black empty] (36) at (6.5, -1) {};
		\node [style=small black empty] (37) at (6.5, 1) {};
		\node [style=small black empty] (42) at (14.25, 0.25) {};
		\node [style=small black empty] (43) at (15.25, 0.25) {};
		\node [style=small black empty] (44) at (14.25, -0.75) {};
		\node [style=none] (45) at (11.5, 3.25) {$S_{all} \Join S_{all}$};
		\node [style=small black empty] (52) at (13.25, 0) {$ti$};
		\node [style=small black empty] (53) at (14.25, -0.25) {};
		\node [style=small black empty] (54) at (15.25, -0.25) {};
		\node [style=small black empty] (55) at (14.25, -1.25) {};
		\node [style=small black empty] (62) at (19.75, 0.25) {};
		\node [style=small black empty] (63) at (20.75, 0.25) {};
		\node [style=small black empty] (64) at (19.75, -0.75) {};
		\node [style=none] (65) at (17, 3.25) {$S_1 \Join S_2$};
		\node [style=small black empty] (72) at (18.75, 0) {$ti$};
		\node [style=small black empty] (73) at (19.75, -0.25) {};
		\node [style=small black empty] (74) at (20.75, -0.25) {};
		\node [style=small black empty] (75) at (19.75, -1.25) {};
		\node [style=small black empty] (76) at (-6.25, 0) {$ti$};
		\node [style=small black empty] (77) at (-5.25, 0) {};
		\node [style=small black empty] (78) at (-4.25, 0) {};
		\node [style=small black empty] (79) at (-5.25, -1) {};
		\node [style=none] (80) at (-7, 1) {$T$};
		\node [style=small black empty] (81) at (-7.5, 0) {$t$};
		\node [style=small black empty] (82) at (-8, 1) {$t1$};
		\node [style=small black empty] (83) at (-8, -1) {$t2$};
		\node [style=small black empty] (84) at (-8, 2) {};
		\node [style=small black empty] (85) at (-3.25, 0) {$t$};
		\node [style=small black empty] (101) at (11.75, 0.75) {};
		\node [style=small black empty] (102) at (10.75, 0.75) {};
		\node [style=small black empty] (103) at (10.75, -0.25) {};
		\node [style=small black empty] (104) at (11.75, 0.75) {};
		\node [style=small black empty] (105) at (11.75, -0.25) {};
		\node [style=small black empty] (106) at (11.75, 1.75) {};
		\node [style=small black empty] (107) at (12.25, -1) {};
		\node [style=small black empty] (108) at (11.25, -1) {};
		\node [style=small black empty] (109) at (11.25, -2) {};
		\node [style=small black empty] (110) at (12.25, -1) {};
		\node [style=small black empty] (111) at (12.25, -2) {};
		\node [style=small black empty] (112) at (12.25, 0) {};
		\node [style=small black empty] (113) at (17.5, 1) {};
		\node [style=small black empty] (114) at (17.5, 0) {};
		\node [style=small black empty] (115) at (17.5, 2) {};
		\node [style=small black empty] (116) at (17.5, -1) {};
		\node [style=small black empty] (117) at (16.5, -1) {};
		\node [style=small black empty] (118) at (16.5, -2) {};
	\end{pgfonlayer}
	\begin{pgfonlayer}{edgelayer}
		\draw [style={blue_edge}] (0) to (3);
		\draw [style={blue_edge}] (0) to (1);
		\draw [style=edge red] (1) to (2);
		\draw [style={blue_edge}] (5) to (8);
		\draw [style={blue_edge}] (5) to (6);
		\draw [style=edge red] (6) to (7);
		\draw [style={blue_edge}] (5) to (10);
		\draw [style=edge red] (10) to (11);
		\draw [style=edge red] (11) to (12);
		\draw [style={blue_edge}] (13) to (16);
		\draw [style={blue_edge}] (13) to (14);
		\draw [style=edge red] (14) to (15);
		\draw [style={blue_edge}] (13) to (18);
		\draw [style={blue_edge}] (18) to (19);
		\draw [style=edge red] (21) to (18);
		\draw [style=edge red] (27) to (28);
		\draw [style=edge red] (28) to (29);
		\draw [style={blue_edge}] (30) to (33);
		\draw [style={blue_edge}] (30) to (31);
		\draw [style=edge red] (31) to (32);
		\draw [style={blue_edge}] (30) to (35);
		\draw [style={blue_edge}] (35) to (36);
		\draw [style=edge red] (37) to (35);
		\draw [style={blue_edge}] (27) to (30);
		\draw [style=edge red] (42) to (43);
		\draw [style={blue_edge}] (52) to (55);
		\draw [style={blue_edge}] (52) to (53);
		\draw [style=edge red] (53) to (54);
		\draw [style={blue_edge}] (52) to (44);
		\draw [style={blue_edge}] (52) to (42);
		\draw [style=edge red] (62) to (63);
		\draw [style={blue_edge}] (72) to (75);
		\draw [style={blue_edge}] (72) to (73);
		\draw [style=edge red] (73) to (74);
		\draw [style={blue_edge}] (72) to (64);
		\draw [style={blue_edge}] (72) to (62);
		\draw [style={blue_edge}] (76) to (79);
		\draw [style={blue_edge}] (76) to (77);
		\draw [style=edge red] (77) to (78);
		\draw [style={blue_edge}] (81) to (76);
		\draw [style=edge red] (81) to (83);
		\draw [style=edge red] (81) to (82);
		\draw [style={blue_edge}] (82) to (84);
		\draw [style={blue_edge}] (85) to (0);
		\draw [style=edge red] (101) to (102);
		\draw [style=edge red] (102) to (103);
		\draw [style={blue_edge}] (104) to (105);
		\draw [style=edge red] (106) to (104);
		\draw [style={blue_edge}] (52) to (104);
		\draw [style=edge red] (107) to (108);
		\draw [style=edge red] (108) to (109);
		\draw [style={blue_edge}] (110) to (111);
		\draw [style=edge red] (112) to (110);
		\draw [style={blue_edge}] (52) to (110);
		\draw [style={blue_edge}] (113) to (114);
		\draw [style=edge red] (115) to (113);
		\draw [style=edge red] (116) to (117);
		\draw [style=edge red] (117) to (118);
		\draw [style={blue_edge}] (113) to (72);
		\draw [style={blue_edge}] (116) to (72);
	\end{pgfonlayer}
\end{tikzpicture}
		\caption{Example of trees for the proof of \Cref{thm:minimal-tree-automata}}
		\label{fig:tree_unfolds}
	\end{figure}

Given a $(R,t) \in K$, we have
$$ R \stackrel {t} \Join (t \arc c t_i) \leq_{subtree} R \stackrel {t} \Join T_i  \in X_i \subset K_i$$
Therefore $(R \stackrel {t} \Join (t \arc c t_i),t_i) \in K_i$,
thus $K \arc c K_i$. However, by induction, $T_i \in K_i$ implies that $\G$ accepts $T_i$ starting on $K_i$, thus, putting everything together, we get that $\G$ accepts $T$ starting on $K$. 

Now, it remains to show that $\G$ does not accept any tree outside $L$. If it were the case, consider the minimal tree $T$ (with respect to size) such that $T \in L^\tree(G) \setminus L$, let $s$ be a leaf of $T$ with (without loss of generality) $s \arc{c} t$ and $\phi$ be a mapping of $T$ into $\G$. Let $ K_s = \phi(s)  \arc{c} \phi(t) = K_t $. Then $T- s$ is accepted by $\G$ starting at $K_t$. By minimality of $T$, $T-s \in K_t$. Then $K_s \arc c K_t$ would imply $((s \arc c t) \stackrel {t} \Join (T-s),s ) = (T,s) \in K_s $.

This ends the proof of (1).
	\vspace{5pt}
	
	(2): Let $G$ be a (possibly infinite) graph such that $L^\tree(G) = L$. Let $\cal K' = \set{K'_i}_{i\geq 0}$ be the languages of $L^\star$ corresponding to the trees accepted by each vertex of $G$. They must correspond to compatible sets since any two trees accepted from the same vertex can be glued together. We have $L^\star = \bigcup K_i = \bigcup K'_i$, but the $K_i$ correspond to the minimal such configuration (in terms of the number of sets), thus $|G| \geq |\G|$. Now, if $|G|=|\Gamma|$, since $\cal K$ is finite, by \Cref{lem:max_simplices}, $\cal K' = \cal K$, and thus $G \cong \G$
		\end{proof}
	
	Note that the heart of the proof relies on the fact that trees can be joined together at a specific vertex, which is the fundamental difference from paths. In order to find a canonical graph for a given graph-based path language, we suggest two options. The first is to lift the word language to a tree language (see \Cref{prop:path-to-tree}) and use minimisation there. Alternatively, \Cref{cor:hom-minimal} establishes the existence of a homomorphism-minimal graph for a given language.  
	
	\begin{proposition}\label{prop:path-to-tree}
		For $L$ a graph-based path-language, the union of all graph-based tree-languages whose restriction to paths is $L$ is a (regular) graph-based tree-language.
	\end{proposition}
	\begin{proof}
		The language $\widetilde L$ we construct can be characterised by : $S \cap L^c = \emptyset \Ra S \subset \widetilde L$ for any graph-based tree-language $S$, where $L^c$ is the complement of $L$ within path languages. Note that $L^c$ is regular in the usual sense of word automata. Applying \Cref{thm:duality-regular} (which will be proven shortly) on $L^c$, there exists a graph $\cal G_{L^c}$ such that $S \cap L^c = \emptyset \Ra S \subset L(\cal G_\G)$.
	\end{proof}
	
	\section{Homomorphism duality} \label{sec:morphism-duality}
	In the study of the homomorphism ordering, a duality between two objects is a pair $(F,H)$ that satisfies a relation of the form $F \not \ra G \iff G \ra H$. For graphs, the only duality pair is $(K_2,K_1)$. For planar graphs, Reza Naserasr showed that $(K_3,K_3\times H_5)$ is a duality pair, where $H_5$ is the Clebsch graph. Note that it is not planar.
	Since duality pairs are rare, we also call a duality any "well-behaved" set of obstructions to a mapping, that is a pair $((F_i)_{i\geq 0},H)$ such that $(\forall i, F_i \not \ra G) \iff G \ra H$. The most famous example is $((C_{2k+1})_{k\geq 1},K_2 )$, which corresponds to the well-known fact that a graph is bipartite if it does not contain an odd cycle.
	
	A central result for duality pairs is that, for oriented graphs, every oriented tree (meaning a tree, with an arbitrary orientation on the edges) is in a duality pair \cite{NT00}. This result actually generalises to more general finite structures \cite{EPTT17}. 

\begin{theorem}\label{thm:duality-single-tree}
	If $T$ is an edge-colored, oriented tree (allowing bidirections), then there exists $\cal G_T$ such that, for any edge-colored oriented graph:
	$$T \not \ra G \iff G \ra \cal G_T$$
	This, in terms of languages, means: 
	$T \not \in L^\tree(G) \iff G \ra \cal G_T$
\end{theorem}
We omit this proof because this is a special case of \Cref{thm:duality-regular}, which will be proven shortly.
\vspace{10pt}

	This in particular means that $L^\tree(\cal G_T)$ is the largest graph-based language which does not contain $T$, and that $\cal G_T$ is the minimal automata (from a homomorphism point of view) that accepts that language.
	We can view this as an alternative to a complementation : we just consider the largest graph-based language in $\set{T}^c$. 
For any finite set $\set{T_i}$, we can construct $\cal G_{\set{T_i}} = \prod_i \cal G_{T_i}$ which verifies $ L^\tree(G) \subset \set{T_i}^c \iff G \ra \cal G_{\set{T_i}}$. The \Cref{thm:duality-regular} generalises this to any regular tree language. This result was originally proven by Erdös, Pálvölgyib, Tardif and Tardos in \cite{EPTT17} in the context of arbitrary $\sigma$-structures. %Regular families of forests, antichains and duality pairs of relational structures
We give here a different proof, tailored to our context.
	
\begin{theorem}\label{thm:duality-regular}
	If $\set{T_i}$ is a regular set of trees recognised by the regular tree automata $A = (Q,\Gamma,F)$, then there exists a graph $\cal G_A$ such that for any edge colored oriented graph $G$:
	$$ \set{T_i} \cap L^\tree(G) = \emptyset \iff G \ra \cal G_A$$
\end{theorem}

\begin{proof}
We start by defining an infinite graph of functions mapping each tree-vertex to a neighbour. We write $V(L(A))$ for the set of all vertices from trees in $L(A)$.

$\cal G_A' = 
\begin{cases}
V = \set{ f \in (\text{func: } V(L(A)) \ra V(L(A))) \mid \forall T \in L(A), \forall t \in V(T), f(t) \in N(t) } \\
E_c = \set{(f,g) | \neg \exists s \arc{c} t, f(s)=t \text{ and } f(t)=s}
\end{cases}$

For a tree $T$, let $\prec$ be the degeneracy ordering on $T$, that is, the ordering of the vertices obtained by successively removing the leaves of the tree. It verifies that, for each vertex $s$ there is at most one neighbour $t$ such that $t \succeq s$ ($*$)

For $T \in L(A), s \in V(T)$, we call the local configuration around $s$, the information of: the states of $s$ and its neighbours, with a corresponding rule of $\Gamma$. %and the structure of the poset $\prec$ restricted to $ N(s) \cup \set s$.
For $T' \in L(A), s' \in V(T')$, we say that $T$ and $T'$ are locally similar around $s,s'$ if the local configuration of $T$ around $s$ is equal to the local configuration of $T'$ around $s'$. We write $T' \locsim{s',s} T$, since a local similarity induces an isomorphism between the neighbourhoods, we may add neighbours of $s$ and their isomorphic image in the notation: $T',t' \locsim{s',s}T,t$.
		Note that a characteristic of regular tree languages is that they have a finite set of local configurations.

Let $G$ be a graph such that no $T \in L(A)$ map to it. 
We construct $\phi : V(G) \ra V(\cal G_{A}')$ as follows : 
$$ \phi(u) = \text{func:}\left(s\in T \ \mapsto\  \min_\prec \set{t \in N(s), \neg \left( \exists T'\!,t' \locsim{s',s} T,t, \ \   T'_{s' \!\angle t'} \xrightarrow{s' \mapsto u} G\right)}\right)$$ 

Where $T \xrightarrow{s \mapsto u} G$	means that there is a homomorphism from $T$ to $G$ mapping $s$ to $u$.  
	Intuitively, $\phi$ is a function that indicates, for each tree vertex, the direction in that tree that must lead to an obstruction.
Observe that, since the number of different local configurations are finite, and $\phi$ only considers the configuration around $s$, there is only a finite number of different possibilities for $\phi(u)$, we call the corresponding finite subgraph $\cal G_A \subset \cal G'_A$.

We now show that, in the definition, the minimum is well defined. Suppose the set is empty, then for all neighbours $t_i$ of $s$, there exists $ T^i\!,t^i_i \locsim{s^i,s} T,t_i, $ verifying $T^i_{s^i \!\angle t^i_i} \xrightarrow{s^i \mapsto u} G$.
Since all branches around $s$ are independent, we can find $T'\in L(A)$ such that $T',t_1',\dots \locsim{s',s} T,t_1,\dots$ and for each $i$, $T'_{s' \!\angle t_i'} = T^i_{s^i \!\angle t^i_i}$ , thus $T'_{s' \!\angle t'_i} \xrightarrow{s' \mapsto u} G$
	
 Merging all the homomorphisms together, we get that $T' \ra G$, which is a contradiction, thus, $\phi$ is well defined.
 
 Now, we show that $\phi$ is a homomorphism from $G$ to $\cal G_A$. Again, towards a contradiction, suppose that $u \arc{c} v$ but $\neg (\phi(u) \arc c \phi(v))$. Call them $f = \phi(u)$  and $ g=\phi(v)$. Then there exists $T$ and $s \arc c t \in E(T)$ verifying $f(s)=t$ and $g(t)=s$.
 Recall the property ($*$) of the elimination order in $T$. Without loss of generality, we can assume $t \succeq s$ which implies $r_i \prec s  \preceq t$ for every neighbour $r_i \neq t$ of $s$. 
By definition of  $f(s)=t$, $t$ is minimal, and therefore, for every $r_i$ there exists $T^i\!,r^i_i \locsim{s^i,s} T,r_i,$ such that $T^i_{s^i \!\angle r^i_i} \xrightarrow{s^i \mapsto u} G$. Again, we can construct $T'\in L(A)$ such that $T',t',r_1',\dots \locsim{s',s} T,t,r_1,\dots$ and for any $i$, $T'_{s' \!\angle r_i'} = T^i_{s^i \!\angle r^i_i}$, thus $T'_{s' \!\angle r_i'} \xrightarrow{s' \mapsto u} G$. Furthermore, $s' \arc c t'$ and $u \arc c v$, so we can extend any homomorphism with $t'$ unmapped and $s'$ mapped to $u$, by mapping $t'$ to $v$.
 
  Merging all the homomorphisms together, we get that $T'_{t' \!\angle s'} \xrightarrow{t' \mapsto v} G$, which contradicts $\phi(v)(t') = s'$. 
  \vspace{5pt}
  
	  Finally, we show that no $T \in L(A)$ admits a homomorphism to $\cal G_A$. Suppose otherwise, and consider the function $s \mapsto \phi(s)(s)$. Since it is a function from a tree to itself that maps any vertex to a neighbour, by the pigeonhole principle, there must exist two adjacent vertices $s \arc c t$ that have been sent across the same edge, thus, to each other. But this contradicts the definition of $\phi(s) \arc c \phi(t)$.
  \end{proof}

By applying \Cref{lem:min-outside-L} and \Cref{thm:duality-regular} to the complement of a tree language, we get the following :

\begin{corollary}
	\label{cor:tree-hom-minimal}
	For every graph-based tree-language $L$, there exists a homomorphism-minimal graph $\cal C^{tree}_L$ accepting $L$. Meaning that for any $G$, 
	$$ L^\tree(G) \subset L^\tree(\cal C^{tree}_L) = L \Ra (G \ra \cal C^{tree}_L)$$
\end{corollary} 
 
By choosing the regular set of trees to be a regular set of paths, we obtain the same result for automata over words.

\begin{corollary}\label{cor:hom-minimal}
	For every graph-based path language $L$, there exists a homomorphism-minimal graph $\cal C_L$ accepting $L$. Meaning that for any $G$, 
	$$ L(G) \subset L(\cal C_L) = L \Ra G \ra \cal C_L$$
\end{corollary}

Note that \Cref{cor:hom-minimal}, is in fact a consequence of the existence of a (homomorphism) minimal automata for (trimmed) NFAs. (\Cref{prop:NFA-hom-minimal}, \cite{ADN92})
	Then, for graph-based automata, it is enough to restrict the universal target to initial and terminal states and their bi-directional transitions.
%A note about minimal non-deterministic automata 

\begin{proposition}[Christian Carrez, 1970]\label{prop:NFA-hom-minimal}
	A (trim) automaton $A$ recognises a subset of a rational language $L$ iﬀ there is a homomorphism of $A$ into $\cal C_L$.
\end{proposition} 

\section{Structure from the language}

In this section, we present two results showcasing how the notion of graph-based language sits at the intersection between graph and automata theory. 

The first part compares graph-based languages with languages from reversible automata and studies how one decomposes into the other. We will use the graph structure to get information about the language.

The second part introduces the notion of the remanent language of a (plain) graph, here, we use the language to get information about a graph.

\subsection{Decomposition into reversible automata}
		Reversible automata are a well-studied class of automata. In this section, we compare reversible automata with graph-based automata and show how the two notions interact. We show that any graph-based language can be obtained as the projection of the language of a reversible automaton, and furthermore, the alphabet needed has size at most $\sum_c \Delta_c$, where $\Delta_c$ is the maximum degree in colour $c$. In the literature, similar kinds of decompositions have been studied with local languages, see Medvedev's theorem \cite{M58} and \cite{CP12}.
	
	\begin{definition}
		A reversible automaton is a finite (possibly incomplete) automaton in which each letter induces a partial one-to-one map from the set of states into itself.
	\end{definition}

		We link this notion with graph-based automata corresponding to undirected edge-colored graphs.
	
	\begin{proposition}\label{prop:reversible-proper}
		For an undirected edge-colored graph $G$, the following are equivalent:
		\begin{itemize}
			\item $G$ corresponds to a reversible automaton
			\item $G$ is properly edge-colored
		\end{itemize}
	\end{proposition}
	\begin{proof}
		Having two edges of colour $c$ sharing a vertex is forbidden in both notions. Conversely, a proper edge colouring will ensure that, for a given colour, the out- and in-degree of every vertex are either both 0 or both 1.
	\end{proof}
	
	\begin{theorem}\label{thm:reversible-decomposition}
		If $G$ is an undirected edge-colored graph, with the maximum degrees in colour $c$ being $\Delta_c$, then $L(G)$ is the direct image of a reversible language $R$ under a letter-to-letter morphism. Furthermore, $R$ uses an alphabet of size at most $\sum_c \Delta_c$.
	\end{theorem}
	
	\begin{proof}
		Consider the graph $G$ with the colour partition $\cal P$ of its edges. We will refine $\cal P$ into $\cal P'$ in a way that the degree of each vertex in $\cal P'$ is at most 1. Thus $\cal P'$ will correspond to a proper colouring of the underlying graph of $G$ and therefore to a reversible language. The associated letter-to-letter morphism is the one mapping every colour of $\cal P'$ to its original colour in $\cal P$.
		
		Now, we show how to get $\cal P'$ from $\cal P$. Consider a colour $c$ and $G_c$, the subgraph induced by the edges of colour $c$. By definition, $G_c$ has maximum degree $\Delta_c$. Using Vizing's theorem, there exists a $\Delta_c$ or $\Delta_c + 1$ proper edge-colouring of $G_c$. We call the newly coloured graph $G_c'$. Furthermore, if $G$ is bipartite, then $G_c$ is also bipartite and, by König's Line colouring theorem, we have a $\Delta_c$ proper edge colouring. 
		
		Suppose $G$ is bipartite. Then, putting all $G_c'$s together and using distinct colours every time, we get a graph $G'$ with $\sum_c \Delta_c$ colours that is properly edge-colored. Therefore, $L(G')$ is reversible.
		
		If $G$ is not bipartite, then consider the double cover of $G$, denoted $G \times K_2$, whose vertices are the $\set{(u,\eps) \mid u \in V(G), \epsilon \in \set{0,1}}$, and verifying $(u,\eps) \stackrel{c}\adj (v,\mu)$ if $u \stackrel{c} \adj  v$ and $\eps \neq \mu$. The double cover of $G$ is a cover of $G$ in the sense of \Cref{def:double-cover}, therefore $L(G) = L(G \times K_2)$ but $G \times K_2$ is bipartite, therefore the previous argument applies.
		%Using \Cref{prop:reversible-proper}, it is enough to refine the coloring of $G$ into a graph 
	\end{proof}
	
	This kind of decomposition is useful because it provides an unambiguous language for describing paths in a graph (or more precisely, paths in the language of the graph).
	The \Cref{fig:Fano_star} shows an 8-edge-colored graph $G$, whose maximum degree in color \emph{a} is 7. From \Cref{thm:reversible-decomposition}, we would expect that a decomposition of $L(G)$ into reversible languages to use an alphabet of size $14$. However, there exists a graph, $H$, with the same language as $G$ (when \emph{a,a',a''} are interpreted as \emph{a}), that admits a decomposition into reversible languages using only an alphabet of size $10$. $H$ is constructed based on the line-point incidence graph of the Fano plane.
	
	% The picture with a rainbow 7-star on the left and a Fano plane incidence graph on the right. 
	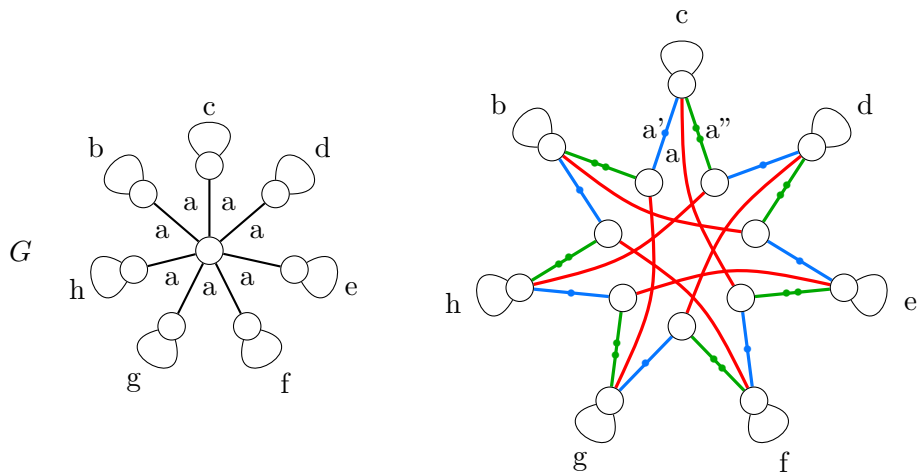
\begin{figure}[ht]
		\centering
		\scalebox{1}{\begin{tikzpicture}
	\begin{pgfonlayer}{nodelayer}
		\node [style=white vertex] (0) at (-7, 2.25) {};
		\node [style=white vertex] (1) at (-7, 0) {};
		\node [style=white vertex] (2) at (-8.75, 1.5) {};
		\node [style=white vertex] (3) at (-5.25, 1.5) {};
		\node [style=white vertex] (5) at (-4.75, -0.5) {};
		\node [style=white vertex] (7) at (-6, -2) {};
		\node [style=white vertex] (8) at (-9, -0.5) {};
		\node [style=white vertex] (9) at (-8, -2) {};
		\node [style=white vertex] (10) at (2.06, 2.74) {};
		\node [style=white vertex] (16) at (5.5, 4.4) {};
		\node [style=white vertex] (13) at (8.94, 2.74) {};
		\node [style=white vertex] (15) at (9.79, -0.98) {};
		\node [style=white vertex] (14) at (7.41, -3.97) {};
		\node [style=white vertex] (12) at (3.59, -3.97) {};
		\node [style=white vertex] (11) at (1.21, -0.98) {};
		\node [style=white vertex] (17) at (3.55, 0.45) {};
		\node [style=white vertex] (18) at (4.63, 1.8) {};
		\node [style=white vertex] (20) at (6.37, 1.8) {};
		\node [style=white vertex] (19) at (7.45, 0.45) {};
		\node [style=white vertex] (22) at (7.06, -1.25) {};
		\node [style=white vertex] (23) at (5.5, -2) {};
		\node [style=white vertex] (21) at (3.94, -1.25) {};
		\node [style=none] (24) at (-7, 3.75) {c};
		\node [style=none] (25) at (-4, 2.75) {d};
		\node [style=none] (26) at (-3.25, -1) {e};
		\node [style=none] (27) at (-5, -3.5) {f};
		\node [style=none] (28) at (-10, 2.75) {b};
		\node [style=none] (30) at (-9, -3.5) {g};
		\node [style=none] (31) at (-10.5, -1) {h};
		\node [style=none] (33) at (-7.5, 1.25) {a};
		\node [style=none] (35) at (-6, -0.75) {a};
		\node [style=none] (36) at (-6.5, 1.25) {a};
		\node [style=none] (38) at (-8, -0.75) {a};
		\node [style=none] (39) at (-7, -1) {a};
		\node [style=none] (40) at (-8.25, 0.5) {a};
		\node [style=none] (41) at (-5.75, 0.5) {a};
		\node [style=none] (42) at (0.65, 3.88) {b};
		\node [style=none] (43) at (5.5, 6.2) {c};
		\node [style=none] (44) at (10.35, 3.88) {d};
		\node [style=none] (45) at (11.55, -1.39) {e};
		\node [style=none] (46) at (8.19, -5.59) {f};
		\node [style=none] (47) at (2.81, -5.59) {g};
		\node [style=none] (48) at (-0.55, -1.39) {h};
		\node [style=none] (49) at (5.25, 2.5) {a};
		\node [style=none] (50) at (4.75, 3.25) {a'};
		\node [style=none] (51) at (6.5, 3.25) {a''};
		\node [style=none] (52) at (-12, 0) {$G$};
	\end{pgfonlayer}
	\begin{pgfonlayer}{edgelayer}
		\draw [thick] (1) to (0);
		\draw [thick] (1) to (3);
		\draw [thick] (1) to (5);
		\draw [thick] (1) to (7);
		\draw [thick] (1) to (9);
		\draw [thick] (8) to (1);
		\draw [thick] (1) to (2);
		\draw [style={blue_edge}] (10) to (17);
		\draw [style={blue_edge}] (13) to (20);
		\draw [style={blue_edge}] (12) to (23);
		\draw [style={blue_edge}] (22) to (14);
		\draw [style={blue_edge}] (15) to (19);
		\draw [style={blue_edge}] (21) to (11);
		\draw [style=edge red, bend right=15, looseness=1.25] (10) to (19);
		\draw [style=edge red, bend right=15, looseness=1.25] (11) to (20);
		\draw [style=edge red, bend left=15, looseness=1.25] (21) to (15);
		\draw [style=edge red, bend right=15, looseness=1.25] (16) to (22);
		\draw [style=edge red, bend left=15, looseness=1.25] (23) to (13);
		\draw [style=edge red, bend left=15, looseness=1.25] (17) to (14);
		\draw [style=edge red, bend right=15, looseness=1.25] (12) to (18);
		\draw [style=green edge] (11) to (17);
		\draw [style=green edge] (10) to (18);
		\draw [style=green edge] (19) to (13);
		\draw [style=green edge] (21) to (12);
		\draw [style=green edge] (20) to (16);
		\draw [style=green edge] (15) to (22);
		\draw [style=green edge] (14) to (23);
		\draw [style={blue_edge}] (18) to (16);
		\draw [in=135, out=45, loop] (0) to ();
		\draw [in=90, out=0, loop] (3) to ();
		\draw [in=-60, out=30, loop] (5) to ();
		\draw [in=-105, out=-15, loop] (7) to ();
		\draw [in=-75, out=-165, loop] (9) to ();
		\draw [in=150, out=-120, loop] (8) to ();
		\draw [in=-180, out=90, loop] (2) to ();
		\draw [in=90, out=180, loop] (10) to ();
		\draw [in=135, out=45, loop] (16) to ();
		\draw [in=90, out=0, loop] (13) to ();
		\draw [in=30, out=-60, loop] (15) to ();
		\draw [in=-15, out=-105, loop] (14) to ();
		\draw [in=-75, out=-165, loop] (12) to ();
		\draw [in=150, out=-120, loop] (11) to ();
	\end{pgfonlayer}
\end{tikzpicture}}
		\caption{A graph with $\Delta_a=7, |\Sigma| = 8$ with a reversible decomposition of the language using 10 colours instead of 14.}
		\label{fig:Fano_star}
	\end{figure}

\subsection{Remanent language}
	Most of the paper deals with \emph{decorated} graphs. In this section, we discuss how we can get information about an \emph{plain} graph (monochromatic and undirected) by studying all of its possible orientations.
	\begin{definition}
		Given a plain graph $G$, the \emph{remanent language} of $G$ is $\rem (G) = \bigcap_{\vec G} L(\vec G)$ where $\vec G$ spans all possible orientations of $G$.
	\end{definition}
	Recall from \Cref{prop:directed-cycle-necessary} that only acyclic orientations have $L(\vec G) \neq \Sigma^*$, therefore, it is enough to consider acyclic orientations only.
	The remanent language of $G$ is a graph-based language, therefore, there is an oriented graph, $\remgraph (G)$, whose language is the remanent language of $G$.
	
	\begin{figure}[ht]
		\centering
		\scalebox{1}{\begin{tikzpicture}
	\begin{pgfonlayer}{nodelayer}
		\node [style=black empty] (0) at (-8.5, 2) {};
		\node [style=black empty] (1) at (-6.75, 3.25) {};
		\node [style=black empty] (2) at (-5, 2) {};
		\node [style=black empty] (3) at (-7.75, 0) {};
		\node [style=black empty] (4) at (-5.75, 0) {};
		\node [style=black empty] (5) at (-3.75, 2) {};
		\node [style=black empty] (6) at (-2, 3.25) {};
		\node [style=black empty] (7) at (-0.25, 2) {};
		\node [style=black empty] (8) at (-3, 0) {};
		\node [style=black empty] (9) at (-1, 0) {};
		\node [style=black empty] (10) at (-13.5, 2) {};
		\node [style=black empty] (11) at (-11.75, 3.25) {};
		\node [style=black empty] (12) at (-10, 2) {};
		\node [style=black empty] (13) at (-12.75, 0) {};
		\node [style=black empty] (14) at (-10.75, 0) {};
		\node [style=small black empty] (15) at (2, 0) {};
		\node [style=small black empty] (16) at (2.5, 1) {};
		\node [style=small black empty] (17) at (3, 2) {};
		\node [style=small black empty] (18) at (3.5, 1) {};
		\node [style=small black empty] (19) at (4.5, 1) {};
		\node [style=small black empty] (20) at (5, 2) {};
		\node [style=small black empty] (21) at (5.5, 3) {};
		\node [style=small black empty] (22) at (6, 2) {};
		\node [style=small black empty] (23) at (6.5, 1) {};
		\node [style=small black empty] (24) at (4, 2) {};
		\node [style=small black empty] (26) at (7, 2) {};
		\node [style=small black empty] (27) at (7.5, 3) {};
		\node [style=small black empty] (28) at (8, 2) {};
		\node [style=small black empty] (29) at (8.5, 3) {};
		\node [style=small black empty] (30) at (9, 2) {};
		\node [style=small black empty] (31) at (9.5, 3) {};
		\node [style=small black empty] (32) at (10, 4) {};
		\node [style=small black empty] (33) at (10.5, 3) {};
		\node [style=small black empty] (34) at (11, 2) {};
		\node [style=small black empty] (36) at (11.5, 3) {};
		\node [style=small black empty] (37) at (12, 4) {};
		\node [style=small black empty] (38) at (12.5, 3) {};
		\node [style=small black empty] (39) at (13.5, 3) {};
		\node [style=small black empty] (40) at (13, 4) {};
		\node [style=small black empty] (42) at (14, 4) {};
		\node [style=small black empty] (43) at (14.5, 5) {};
		\node [style=small black empty] (44) at (15, 4) {};
	\end{pgfonlayer}
	\begin{pgfonlayer}{edgelayer}
		\draw [style=Arc] (0) to (1);
		\draw [style=Arc] (1) to (2);
		\draw [style=Arc] (0) to (3);
		\draw [style=Arc] (3) to (4);
		\draw [style=Arc] (4) to (2);
		\draw [style=Arc] (5) to (6);
		\draw [style=Arc] (5) to (8);
		\draw [style=Arc] (9) to (7);
		\draw [style=Arc] (6) to (7);
		\draw [style=Arc] (9) to (8);
		\draw [style=Arc] (10) to (11);
		\draw [style=Arc] (10) to (13);
		\draw [style=Arc] (13) to (14);
		\draw [style=Arc] (14) to (12);
		\draw [style=Arc] (12) to (11);
		\draw [style=Arc] (15) to (16);
		\draw [style=Arc] (16) to (17);
		\draw [style=Arc] (18) to (17);
		\draw [style=Arc] (19) to (20);
		\draw [style=Arc] (20) to (21);
		\draw [style=Arc] (22) to (21);
		\draw [style=Arc] (23) to (22);
		\draw [style=Arc] (26) to (27);
		\draw [style=Arc] (28) to (27);
		\draw [style=Arc] (23) to (26);
		\draw [style=Arc] (30) to (31);
		\draw [style=Arc] (31) to (32);
		\draw [style=Arc] (33) to (32);
		\draw [style=Arc] (34) to (33);
		\draw [style=Arc] (28) to (29);
		\draw [style=Arc] (30) to (29);
		\draw [style=Arc] (18) to (24);
		\draw [style=Arc] (19) to (24);
		\draw [style=Arc] (36) to (37);
		\draw [style=Arc] (38) to (37);
		\draw [style=Arc] (38) to (40);
		\draw [style=Arc] (39) to (40);
		\draw [style=Arc] (34) to (36);
		\draw [style=Arc] (42) to (43);
		\draw [style=Arc] (44) to (43);
		\draw [style=Arc] (39) to (42);
		\draw [style=Arc, in=-90, out=0, looseness=0.75] (15) to (44);
	\end{pgfonlayer}
\end{tikzpicture}}
		\caption{The acyclic orientations of $C_5$ (left) and $\remgraph (C_5)$ (right) }
		\label{fig:rem_ex}
	\end{figure}
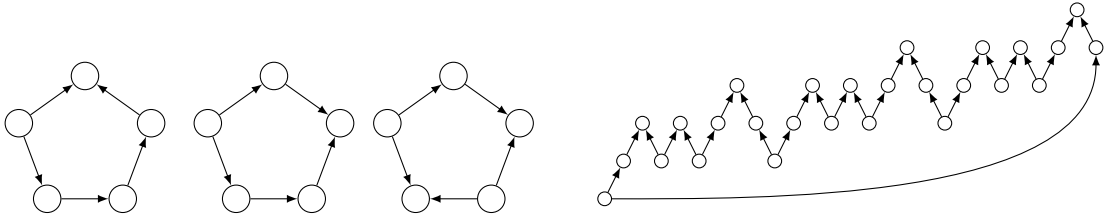
	
	The main motivation for remanent languages is the following property.
	
	\begin{proposition}\label{prop:rem-duality}
	If $(P, \cal G_P)$ is a duality pair (as in \Cref{thm:duality-single-tree}) between an oriented path and a graph, with $U(\cal G_P)$ is the underlying graph of $\cal G_P$ :
	$$P \in \rem(G) \iff G \not \ra U(\cal G_P)$$
	\end{proposition}
	\begin{proof}
	For $\vec G$ an orientation of $G$, the definition of a duality gives us: $P \in L(\vec G) \iff {\vec G \not \ra \cal G_P}$. By intersecting over all possible $\vec G$, we get that $P \in \rem(G)$ if and only if no orientation of $G$ can map to $\cal G_P$, which is equivalent to the fact that $G$ does not map to $U(\cal G_P)$.
	\end{proof}
	Using the duality between the directed path $\vec P_k = \tra^{k-1}$ and the transitive tournament $TT_k$, we get:
	\begin{corollary}
		For a plain graph $G$ and integer $k$
		$$\vec P_k \in \rem(G) \iff \chi(G) \geq k$$
	\end{corollary}
	
	\Cref{prop:rem-duality} raises the question of which undirected graphs correspond to some $U(\cal G_P)$.
	The question for tree families was studied by Hell, Nešetřil and Zhu in \cite{HNZ96}. 
	And in \cite{SC21}, it is proved that some orientation of the odd cycle $C_{2k+1}$ was in duality pair with the oriented path $\tra (\tra \tla)^{k-1}\tra\tra$.
	
	The height of an oriented path or cycle, denoted $h(\bullet)$, is the (absolute) difference between the number of forward and backward edges. 	
	
	\begin{proposition}\label{prop:big-height-rem}
		There are oriented paths of arbitrarily large height in $\rem(C_{2k+1})$.
	\end{proposition}
	\begin{proof}
		We start by recalling the following fact :
		\begin{quote}
			If $C$ and $C'$ are two oriented odd cycles, then $C \times C'$ contains an oriented odd cycle.
		\end{quote}
		This is a consequence of $f(3)=3$ where $f$ is the Poljak–Rödl function for oriented graphs \cite[Thm 3.3]{SR81} %\Cyril{ On the Arc-Chromatic Number of a Digraph, Thm 3.3}.
		 Iterating the result, we get that there is an odd cycle $C_0$ in the product $\prod \vec C$ where $\vec C$ ranges over all orientations of $C_{2k+1}$.
		Since $C_0$ is odd, it must have odd, and thus non-zero, height. Let $P_0$ be an oriented path that winds around $C_0$ an arbitrary number of times. Then $P_0$ has arbitrarily large height and satisfies $P_0 \ra C_0 \ra \prod \vec C$. Thus, $P_0 \in \rem(C_{2k+1})$.
	\end{proof}
	
	 We say that $\rem(G)$ is trivial if $\rem(G) \subset L(P)$ for some oriented path $P$.
	
	\begin{theorem}
		$\rem(G)$ is trivial if and only if $\rem(G) = L(\vec K_2)$ if and only if $G$ is bipartite.
	\end{theorem}
	\begin{proof}
		If $G$ contains no arc, then the result is trivial, so we assume $G$ to contain an arc, in particular $L(\vec K_2) \subset \rem(G)$.
		
		If $G$ is bipartite, then $G \ra K_2$, and there exists an orientation $\vec G$ of $G$ such that $\vec G \ra \vec K_2$. Then $\rem(G) \subset L(\vec G) \subset L(\vec K_2) \subset \rem(G)$. We have that $\rem(G)$ is trivial.
		
		If $G$ is non-bipartite, then it contains an odd cycle, so by \Cref{prop:big-height-rem}, $\rem(G)$ contains oriented paths of arbitrarily large height. But, observe that homomorphisms preserve height, in the sense, that if $\phi : P \ra P'$ for two oriented paths $P$ and $P'$, then $h(P) = h(\phi(P))$. So $L(P)$ contains only oriented paths of bounded height. Thus $\rem(G)$ cannot be trivial.
	\end{proof}

	\paragraph{Acknowledgements:} I would like to thank Marie Fortin and Daniela Petrisan   for their guidance and exchanges. I would also like to thank Killian Barbé, Samuel Coulomb and Luc Passemard for their interest and the discussions we had together.	

%\clearpage
\bibliographystyle{plain}
\bibliography{biblio}
	
\end{document}